\documentclass[9pt,twocolumn,twoside %,lineno
]{gsajnl}
\usepackage{epstopdf}
\articletype{inv} % article type
\runningtitle{Evaluation of  model fit } % For use in the footer
\runningauthor{van Waaij \textit{et al.}}

\title{Evaluation of population structure inferred by principal component analysis or the admixture model}

\author[1,3,$\dagger$]{Jan van Waaij}
\author[1,$\dagger$]{Song Li}
\author[2]{Gen\'{i}s Garcia-Erill}
\author[2]{Anders Albrechtsen}
\author[1,$\ast$]{Carsten Wiuf}

\affil[1]{Department of Mathematical Science, University of Copenhagen, 2100 Copenhagen, Denmark}
\affil[2]{Department of Biology, University of Copenhagen, 2100 Copenhagen, Denmark}
\affil[3]{Current address: Department of Health Technology, Danish Technical University, 2800 Kgs. Lyngby, Denmark}
\affil[$\dagger$]{These authors contributed equally to this work}

\correspondingauthoraffiliation[$\ast$]{Corresponding author: Department of Mathematical Sciences, Universitetsparken 5, 2100 Copenhagen, Denmark. Email: wiuf@math.ku.dk}
\begin{abstract}
Principal component analysis (PCA) is commonly used in genetics to infer and visualize population structure and admixture between populations. PCA is often interpreted in a way similar to inferred admixture proportions, where it is assumed that individuals belong to one of several possible populations or are admixed between these populations. We propose a new method to assess the statistical fit of PCA (interpreted as a model spanned by the top principal components) and to show that violations of the PCA assumptions affect the fit. Our method uses the chosen top principal components to predict the genotypes. By assessing the covariance (and the correlation) of the residuals (the differences between observed and predicted genotypes), we are able to detect violation of the model assumptions. Based on simulations and genome wide human data we show that our assessment of fit can be used to guide the interpretation of the data and to pinpoint individuals that are not well represented by the chosen principal components. Our method works equally on other similar models, such as the admixture model, where the mean of the data is represented by linear matrix decomposition.

\end{abstract}


\dates{\rec{xx xx, xxxx} \acc{xx xx, xxxx}}

\usepackage{mathptmx}

\usepackage{graphics}
\usepackage{wrapfig}
\usepackage{graphicx}
\usepackage{float}
\usepackage{subfigure}
\usepackage{overpic}
\usepackage{cuted}
\usepackage{extarrows}

\usepackage{epsfig}

\usepackage{caption}

\usepackage{amssymb,amsfonts,amsmath,amsthm}
\usepackage{xcolor,tikz,pgf}

\usepackage{booktabs}
\usepackage{mathtools}
\usepackage{soul}
\usepackage{tablefootnote} % for table footnotes

\usepackage{bm}  % Define \bm{} to use bold math fonts
\usepackage{bbm}  % Define \bm{} to use bold math fonts
\usepackage{thm-restate}
\usepackage{fancyhdr}
\usepackage{float}
\usepackage{hyperref}

\usepackage[utf8]{inputenc}
\makeatletter
\newcommand{\refcheckize}[1]{%
  \expandafter\let\csname @@\string#1\endcsname#1%
  \expandafter\DeclareRobustCommand\csname relax\string#1\endcsname[1]{%
    \csname @@\string#1\endcsname{##1}\wrtusdrf{##1}}%
  \expandafter\let\expandafter#1\csname relax\string#1\endcsname
}
\makeatother
\newcommand*\xbar[1]{%
  \hbox{%
    \vbox{%
      \hrule height 0.5pt % The actual bar
      \kern0.3ex%         % Distance between bar and symbol
      \hbox{%
        \kern-0.1em%      % Shortening on the left side
        \ensuremath{#1}%
        \kern 0em%      % Shortening on the right side
      }%
    }%
  }%
}

\newtheorem{theorem}{Theorem}

\newtheorem{lemma}[theorem]{Lemma}

\DeclareMathOperator{\R}{\mathbb{R}}
\DeclareMathOperator{\re}{\mathbb{R}}
\DeclareMathOperator{\E}{\mathbb{E}}

\newcommand{\rh}[1]{\left(#1\right)}
\newcommand{\vh}[1]{\left[ #1 \right]}

\newcommand{\inpr}[1]{\left\langle#1\right\rangle}

\DeclareMathOperator{\cov}{cov}

\DeclareMathOperator{\var}{var}

\DeclareMathOperator{\x}{\times}
\DeclareMathOperator{\diag}{diag}

\DeclareMathOperator{\sdot}{\,\cdot\,}

\DeclareMathOperator{\binomial}{binomial}
\renewcommand{\sp}{span}

\usepackage[colorinlistoftodos]{todonotes}
\begin{document}
%\pagenumbering{arabic}
%% Do NOT include any frontmatter information; including the title, author names,
%% institutes, acknowledgments and title footnotes (author information, funding
%% sources, etc.). Start the document with the first section or paragraph of
%% the article.

%genetics version
\maketitle
%\thispagestyle{firststyle}
%\slugnote
%\firstpagefootnote
\vspace{-13pt}% Only used for adjusting extra space in the left column of the first page

\section{Introduction}

Principal component analysis (PCA) and model-based clustering methods are popular ways to disentangle the ancestral genetic history of individuals and populations. One particular model, the admixture model \citep{Pritchardea2000}, has played a prominent role because of its simple structure and, in some cases, easy interpretability. PCA is often seen as being model free but as noted by \cite{engelhardtstephens2010}, the two approaches are very similar. The interpretation of the results of a PCA analysis is often  based on  assumptions similar to those of the admixture model, such that admixed individuals are linear combinations of the eigenvectors  representing unadmixed individuals. In this way, the admixed individuals lie in-between the unadmixed individuals in a PCA plot. As shown for the admixture model, there are many demographic histories that can lead to the same result \citep{LawsonvanDorpFalush2018} and many demographic histories that violate the assumptions of the admixture model \citep{genisanders2020}. As we will show, this is also the case for PCA, since it has a similar underlying model \citep{engelhardtstephens2010}.

The admixture model states that the genetic material from each individual is composed of contributions from $k$ distinct ancestral homogeneous populations. However, this is often contested in real data analysis, where the ancestral population structure might be much more complicated than that specified by the admixture model. For example, the $k$ ancestral populations might be heterogeneous themselves, the exact number of ancestral populations might be difficult to assess due to many smaller contributing populations, or the genetic composition of an individual might be the result of continuous migration or recent backcrossing, which also violates the assumptions of the admixture model. Furthermore, the admixture model assumes individuals are unrelated, which naturally might not be the case.   This paper is concerned with assessing the fit of PCA building on the special relationship with the admixture model \citep{engelhardtstephens2010}. In particular, we are interested in quantifying the model fit and assessing the validity of the model at the level of the sample as well as at the level of the individual. Using real and simulated data we show that the fit from a PCA analysis is affected by violations of the admixture model.

We consider genotype data $G$ from $n$ individuals and $m$ SNPs, such that $G_{si}\in\{0,1,2\}$ is the number of reference alleles for individual $i$ and SNP $s$. Typically, $G_{si}$ is assumed to be binomially distributed with parameter $\Pi_{si}$, where $\Pi_{si}$ depends on the number of ancestral populations, $k$, their admixture proportions and the ancestral population allele frequencies. For clustering based analysis such as ADMIXTURE \citep{alexander-lange}, $k$ is the number of clusters while in  PCA, it is the $k-1$ top principal components.  We give the specifics of the admixture model in the next section and show its relationship to PCA in the Material and methods section. 
%\cref{subsec:model}.  
 
Several methods aim to estimate the best $k$ in some sense \citep{alexander-lange,evanno,Pritchardea2000,raj,wang2019}, but finding such $k$ does not imply the data fit the model   \citep{lawson,janes}.
In statistics, it is standard  to use residuals and distributional summaries of the residuals to assess    model fit  \citep{box2005}. The residual of an observation is defined as the difference between the observed and the predicted value (estimated under some model).
Visual trends in the residuals (for example, differences between populations) are indicative of model misfit, and large absolute values of the residuals are indicative of  outliers (for example due to experimental errors, or kinship). If the model is correct, a histogram of the residuals is expected to be mono-modal centered around zero  \citep{box2005}. 

In our context, \cite{genisanders2020}   argue  that   trends in the residual correlation matrix   carries information about the underlying model and might be used for visual model evaluation.  A method is designed to  assess whether the correlation structure agrees with the proposed model, in particular, whether it agrees with the proposed number of homogeneous ancestral populations    \citep{genisanders2020}. However, even in the case the model is correctly specified,   the residuals are in general correlated \citep{box2005}, and therefore,  trends might be observed even if the model is true, leading to incorrect model assessment. To adjust for this correlation,  a leave-one-out procedure, based on maximum likelihood estimation of the admixture model parameters, is developed that removes the correlation between residuals in the case the model is correct, but not if the model is misspecified \citep{genisanders2020}. This approach could also be applied to PCA, where %missing genotypes could be introduced and
expected genotypes could be calculated using probabilistic PCA \citep{meisner2021}. 
This leave-one-out procedure is, however, computationally expensive. 

To remedy the computational difficulties, we take a different approach  to investigate the correlation structure.
We suggest two different ways of calculating the correlation  matrix of the residuals. The first is simply the empirical correlation matrix of the residuals. The second might be considered an estimated correlation matrix, based on a model. Both are simple  to compute. Under mild regularity assumptions, these two measures agree if the model is correct and the number of SNPs is large. Hence, their difference is expected to be close to zero, when the admixture model is not violated. If the difference is considerably different from zero, then this is proof of model misfit.

To explore the adequacy of the proposed method, we investigate different ways to calculate the predicted values of the genotype (hence, the residuals), using Principal Component Analysis (PCA) in different ways. However, we also show that this approach can be used on   estimated admixture proportions.  
 Specifically, we use  1) an uncommon but very useful PCA approach (here, named PCA 1) based on  unnormalized genotypes  \citep{CabrerosStorey2019,chenstorey2015}, 2)  PCA applied to  mean centred data (PCA 2), see \cite{Patterson2006}, and 3) PCA applied to mean and variance normalised data (PCA 3) \citep{Patterson2006}.   All three approaches are computationally fast and do  not require separate estimation of ancestral allele frequencies and population proportions, as in \cite{genisanders2020}. Hence, the computation of the residuals are computationally inexpensive. Additionally, we show that this approach can also be applied to output from, for example, the software  ADMIXTURE  \citep{Alexander2009} to estimate $\Pi_{si}$ for each $s$ and $i$, and to calculate the residuals from these estimates. An overview of PCA can be found in \cite{Jolliffe2022}.

We demonstrate  that our proposed method works well on simulated and real data, when the predicted values (and the residuals) are calculated in any of the four mentioned ways. Furthermore, we back this up mathematically by showing that the two correlation measures agree (if the number of SNPs is large) under the correct admixture model for PCA 1 and PCA 2. 
For the latter, a  few additional assumptions are required.  The estimated covariance (and correlation coefficient) under the proposed model might be seen as a correction term  for population structure. Subtracting it from the empirical covariance, thus gives a covariance estimate with baseline zero under the correct model, independent of the population structure. It is natural to suspect that similar can be done in models with population structure and kinship, which we will pursue in a subsequent study.

  In  the next section, %\ref{sec:mm}, 
  we describe the model, the statistical approach to compute the residuals, and how we evaluate model fit. In addition, we give mathematical statements that show how the method performs theoretically. In the `Results' section, %\Cref{subsec:sim}, 
  we provide analysis of simulated and real data, respectively.   We end  with a discussion. % in \Cref{sec:discussion}.   
  Mathematical proofs are collected in the appendix.

\section{Materials and methods}%\label{sec:mm}

\subsection{Notation}

For an $\ell_1\times\ell_2$ matrix $A=(A_{ij})_{i,j}$, $A_{\star i}$ denotes the $i$-th column of $A$,  $A_{i \star }$ the $i$-th row, $A^T$ the transpose matrix, and  $\text{rank}(A)$ the rank.
The   Frobenius norm of a square $\ell\times\ell$ matrix $A$ is   
$$\|A\|_F=\sqrt{\sum_{i=1}^\ell\sum_{j=1}^\ell A_{ij}^2}.$$
A square matrix $A$ is an orthogonal projection if $A^2=A$ and $A^T=A$. A symmetric matrix has $n$ real eigenvalues (with multiplicity) and the eigenvectors can be chosen such that they are orthogonal to each other. If the matrix is positive (semi-)definite, then the eigenvalues are positive (non-negative).

For a random variable/vector/matrix  $X$,   its expectation  is denoted   $\E[X]$  (provided it exist).  The variance of a random variable $X$ is denoted $\var(X)$, and covariance  between two random variables 	$X,Y$ is denoted $\cov(X,Y)$ (provided they  exist). Similarly, for a random vector $X=(X_1,\ldots,X_n)$,  the covariance matrix is denoted $\cov(X)$.    For a sequence $X_m$, $m=0,\ldots,$ of random variables/vectors/matrices, if $X_m\to X_0$ as $m\to\infty$ almost surely (convergence for all realisations but a set of zero probability),  we leave out `almost surely'  and write $X_m\to X_0$  as $m\to\infty$ for convenience.

\subsection{The PCA and the admixture model}\label{subsec:model}

We consider a model  with genotype observations from $n$ individuals, and $m$ biallelic sites (SNPs), where $m$ is assumed to be (much) larger than $n$, $m\ge n$. The genotype
$G_{si}$ of SNP $s$ in individual $i$ is assumed to be a binomial random variable
\begin{equation*}\label{eq:model}
G_{si}\sim \binomial(2,\Pi_{si}). 
\end{equation*}
In matrix notation, we have  $G\sim\binomial(2,\Pi)$ with expectation $\E( G\mid \Pi)  = 2\Pi $, where $G$ and $\Pi$ are $m\times n$ dimensional matrices. Conditional on $\Pi$, we assume the entries of $G$ are independent random variables.  

Furthermore, we assume the matrix $\Pi$ takes the form $\Pi=FQ$, where $Q$ is a (possibly unconstrained) $k\x n$ matrix of rank $k\le n$, and $F$ is a (possibly unconstrained) 
$m\times k$ matrix, also of rank $k$ (implying $\Pi$ likewise is of rank $k$, Lemma \ref{lem:rankk}). Entry-wise, this amounts to
\[\Pi_{si}=(FQ)_{si}=\sum_{k=1}^kF_{sk}Q_{ki},\quad s=1,\dots,n,\quad i=1,\ldots,m. \]
For the binomial assumption to make sense, we must  require the entries of $\Pi$   to be between zero and one.

In the literature, this model is typically encountered in the form of an admixture model
with $k$ ancestral populations, see for example, \cite{Pritchardea2000,genisanders2020}. The general unconstrained setting which applies to PCA has also been  discussed \citep{CabrerosStorey2019}.  In the case of an admixture model, $Q$ is a matrix of ancestral admixture proportions, such that the proportion of   individual $i$'s genome originating from population $j$ is $Q_{ji}$.   Furthermore, $F$ is a matrix of ancestral SNP frequencies, such that the frequency of the reference allele of SNP $s$ in population $j$ is $F_{sj}$.  In many applications, the columns of $Q$ sum to one.

While we lean towards an interpretation  in terms of  ancestral population proportions and SNP frequencies, our approach does not enforce or assume the columns of $Q$ (the admixture proportions)  to sum to one, but  allow these to be unconstrained.  This is  advantageous for at least two reasons. First, a proposed model might only contain the major ancestral populations, leaving out older or lesser defined populations. Hence, the sum of ancestral proportions might    be smaller than one. Secondly,  when fitting a model with fewer ancestral populations than the true model, one should only require the admixture proportions to sum to at most one.

\subsection{The residuals}%\label{sec:residualmethod}
 
Our goal is to design a strategy to assess the hypothesis that   $\Pi$ is a product of two matrices.
As we do not know  the  true $k$, we suggest a number $k'$ of ancestral populations and estimate the model parameters under this constraint. That is, we assume a model of the form
\begin{equation*}%\label{eq:model_k}
G\sim\binomial(2,\Pi_{k'}),\quad \Pi_{k'}=F_{k'}Q_{k'}, 
\end{equation*} 
where  each entry of $G$ follows a binomial distribution. $ Q_{k'}$ has dimension $k'\x n$,     $F_{k'}$  has dimension  $m\x k'$, and $\text{rank}(Q_{k'})=\text{rank}(F_{k'})=k'$, hence also $\text{rank}(\Pi_{k'})=k'$.   Throughout, we use the index $k'$ to indicate the imposed rank condition, and assume $k'\le k$ unless otherwise stated. The latter assumption is only to guarantee the mathematical validity of certain statements, and is not required for practical use of the method.

Our approach is build on the residuals, the difference between observed and predicted data. To define the residuals, we let  $P\colon \R^n\to\R^n$ be the orthogonal projection onto the $k$-dimensional subspace spanned by the $k$ rows of (the true) $Q$,  hence $P=Q^T(QQ^T)^{-1}Q$, and  $QP=Q$. Let $\widehat P_{k'}$ be an estimate of   $P$ based on the data $G$, and assume $\widehat P_{k'}$  is an orthogonal projection onto a $k'$-dimensional subspace. 
%In \Cref{subsec:estPi},
Later in this section, we  show how an estimate $\widehat P_{k'}$ can be obtained from an  estimate of $Q_{k'}$ or an  estimate of $\Pi_{k'}$. Estimates of these parameters might be obtained using existing methods, based on for example, maximum likelihood analysis \citep{wang2003,Alexander2009,genisanders2020}. Furthermore, for the three PCA approaches, an estimate of the projection  matrix can simply be obtained from   eigenvectors of a singular value decomposition (SVD) of the data matrix. %, see  \Cref{subsec:estPi}.

We define the $m\times n$ matrix of residuals by
$$R_{k'} = G-2\widehat{\Pi} =G(I-\widehat P_{k'}),$$
where $G$ is the observed data and $G\widehat P_{k'}$, the predicted values. The latter might also be considered an  estimate of $2\Pi$, the expected value of $G$.
This definition of residuals is in line with how the residuals are defined in a multilinear regression model  as the difference between the observed data (here, $G$) and the projection of the data onto the subspace spanned by the regressors (here, $G\widehat P_{k'}$). The   essential difference being that  in a multilinear regression model, the regressors  are known and does not depend on the observed data, while   $\widehat P_{k'}$ is   estimated from the data. 

 We  assess the model fit by studying the correlation matrix of the residuals  in two ways.
First, we consider the \emph{empirical covariance matrix } $\widehat B$ with entries
\begin{align*}
\widehat B_{ij}&=\frac 1{m-1}\sum_{s=1}^m (R_{k',si}-{\xbar R}_{k',i} )(R_{k',sj}-{\xbar R}_{k',j})\\
&=\frac 1{m-1}\sum_{s=1}^m (R_{k',si}R_{k',sj}-{\xbar R}_{k',i}\ {\xbar R}_{k',j}),
\end{align*}
 where \[
{\xbar R}_{k'i} = \frac1m \sum_{s=1}^m R_{k',si},
\]
and the corresponding \emph{empirical correlation matrix} with entries
\[
\widehat b_{ij}=\frac{\widehat B_{ij}}{\sqrt{\widehat B_{ii}\widehat B_{jj}}},
\]
$i,j=1,\ldots,n$.
Secondly,  we consider the \emph{estimated covariance matrix} 
 \[
\widehat C = (I-\widehat P_{k'})\widehat D(I-\widehat P_{k'})
 \]
with corresponding  \emph{estimated correlation matrix},
 \begin{align*}
\widehat c_{ij} = \frac{\widehat C_{ij}}{\sqrt{\widehat C_{ii}\widehat C_{jj}}},
\end{align*}
$i,j=1,\ldots,n$.
Here, $\widehat D$ is the $n\times n$ diagonal matrix containing the  average heterozygosities of each individual,
 \begin{equation*}
\label{eq:hatD}
\widehat D_{ii}= \frac1m\sum_{s=1}^m G_{si}(2-G_{si}), \quad i=1,\ldots, n.
\end{equation*}

Under reasonable regularity conditions, we can quantify the behaviour of $\widehat B$ and $\widehat C$ as the number of SNPs become large. Specifically, we assume the rows of $F$ are independent and identically distributed with distribution  $\text{Dist}(\mu, \Sigma)$, where $\mu$ denote the $k$-dimensional mean vector of the distribution, and $\Sigma$ the $k\x k$-covariance matrix, that is,
\begin{align*}
    F_{s\star}=(F_{s1},\ldots,F_{sk})\,  &\stackrel{\text{iid}}\sim  \,\text{Dist}(\mu,\Sigma),  
\end{align*}
$s=1,\ldots,m$.
The matrix $Q$ is assumed to be non-random, that is, fixed.   These assumptions are standard and typically used in simulation of genetic data, see for example, \cite{PickrellPritchard2012,CabrerosStorey2019,genisanders2020}.  Often $\text{dist}(\mu, \Sigma)$ is taken to be the product of $k$ independent uniform distributions in which case $\mu=0.5(1 ,1 ,\ldots ,1)$ and $\Sigma$ is a diagonal matrix with entries $1/12$, though other choices have been applied, see for example \citet{Balding1995,Conomosea2016}. 

Let $D$ be the diagonal matrix   with entries
\begin{equation}\label{eq:Dvar}
D_{ii}=2\E[\Pi_{si}(1-\Pi_{si})],\quad i=1,\ldots,n.
\end{equation}
It follows from Lemma~\ref{thm:unbiasedD} in the appendix, %$\Cref{sec:math}, 
that $\widehat D $ converges 
 to $D$ as $m\to \infty$. Furthermore, as $D_{ii}$ is the variance of $G_{si}$ (it is binomial), then $\widehat D_{ii}$ might be considered an estimate of this variance.
The proofs of the statements are in the appendix. %\Cref{sec:math}.

\begin{theorem}\label{thm:BC}
 Let $k'\le k$. Under the given assumptions, suppose further that $\widehat P_{k'}\to P_{k'}$ as $m\to \infty$, for some   matrix $P_{k'}$. Then, $P_{k'}$ is an orthogonal projection. 
Furthermore,  the following holds,
\begin{align*}
\widehat B  &\,\to \,  (I-P_{k'})(D+4 Q^T\Sigma Q)(I-P_{k'}), \\
\widehat C &\,\to\, (I-  P_{k'})  D(I- P_{k'}),
\end{align*}
  as $m\to\infty$.  Hence, also
  \begin{align*}
\widehat B\,-\,\widehat C\,\, &\to\,\,  4(I-  P_{k'})Q^T\Sigma Q(I-  P_{k'}) \\
&\,\,=\,\,4(P-  P_{k'})Q^T\Sigma Q(P-  P_{k'}),
\end{align*}
 as $m\to\infty$.  For $k'=k$, if $P_k=  P$, then the right hand side is the zero matrix, whereas this is not the case in general for $k'<k$.
\end{theorem}

\begin{theorem}\label{thm:n-1}
Assume $k'=k$ and $P_k=  P$. Furthermore, suppose as in Theorem~\ref{thm:BC}  and that the vector  with all entries equal to one  is in the space spanned by the rows of $Q$ (this is, for example, the case if the admixture proportions sum to one for each individual). Then,
\begin{align}\label{eq:-1}
\frac{\sum_{i=1}^n\sum_{j= 1,i\not=j }^n\widehat B_{ij}}{\sum_{i=1}^n \widehat B_{ii}}&\to \ -1,\quad\text{as}\quad m\to\infty.
\end{align}
  In addition, if $Q$ takes the form
$$Q=\begin{pmatrix} Q_1 & 0& \cdots &0\\ 0 & Q_2   &\cdots &0\\ \vdots &\vdots & \ddots & \vdots \\ 0&0&\cdots& Q_r\end{pmatrix}$$
where $Q_\ell$ has dimension  $k_\ell\times n_\ell$, $\sum_{\ell=1}^r k_\ell=k$ and $\sum_{\ell=1}^r n_\ell=n$, then \eqref{eq:-1}  holds   for each component of $n_\ell$ individuals. If $Q_\ell=(1 \ldots 1)$, then 
$$\widehat b_{ij}\ \to\ -\frac{1}{n_\ell-1},\quad\text{as}\quad m\to\infty,$$
for all individuals $i,j$ in the $\ell$-th component, irrespective the form of $Q_{\ell'}$, $\ell'\not=\ell$.
 \end{theorem}

\begin{theorem}\label{thm:sub}
Assume $k'=k$ and $P_k=  P$. Furthermore, suppose as in Theorem~\ref{thm:BC}  and that $Q$ takes the form
$$Q=\begin{pmatrix} Q_1 & Q_2 \\ 0 & Q_3    \end{pmatrix},$$
where $Q_1=(1 \ldots 1)$ has dimension  $ 1\times n_1$, $n_1\le n$. 
Then, $\widehat b_{ij}$ converges as $m\to\infty$ to a value larger than or equal to $-\tfrac{1}{n_1-1},$
for all   $i,j=1,\ldots,n_1$.
 \end{theorem}

The same statements in the last two theorems hold with $\widehat B$ and $\widehat b$  replaced by $\widehat C$ and $\widehat c$, respectively.

The three theorems provide means to evaluate the model. In particular, Theorem~\ref{thm:BC} might be used to assess the correctness (or appropriateness) of the proposed $k'$, while Theorem~\ref{thm:n-1} and Theorem~\ref{thm:sub} might be used to assess whether data from a group of individuals (e.g., a modern day population) originates from a single ancestral population,  irrespective, the origin of the remaining individuals.  We give examples in the Results  section. %\Cref{subsec:sim}.

The work flow is shown in  Algorithm \ref{alg:alg1}.  We process real and simulated genotype data using PCA 1, PCA 2, PCA 3, and the software ADMIXTURE, and evaluate the fit of the model.

\begin{algorithm}
%\SetAlgoLined
\begin{enumerate}
\item Choose $k'$,
\item Compute an estimate $\widehat P_{k'}$ of the projection $P$,
\item Calculate the residuals $R_{k'}=G(I-\widehat P_{k'})$,
\item Calculate the correlation coefficients, $\widehat b$  and  $\widehat c$,
\item  Plot $\widehat b$ and the difference, the corrected correlation coefficients, $\widehat b-\widehat c$,
\item  Assess visually the fit of the model.
\end{enumerate}
 
\caption{Work flow of the proposed method}\label{alg:alg1}
\end{algorithm}

\subsection{Estimation of $P_{k'}$}%\label{subsec:estPi}

Estimation of $Q, F,$ and $ \Pi$   has received considerable interest in the literature, using for example, maximum likelihood \citep{wang2003,Alexander2009}, Bayesian approaches \citep{Pritchardea2000} or PCA \citep{engelhardtstephens2010}.   

We discuss different ways to obtain an estimate $\widehat P_{k'}$ of $P$.

\subsubsection{Using an estimate $\widehat Q_{k'}$ of $Q_{k'}$ } 
 %\label{sec:Q}
An estimate $\widehat P_{k'}$  might be obtained by projecting onto the subspace spanned by the $k'$ rows of $\widehat Q_{k'}$,
$$\widehat P_{k'}=\widehat Q_{k'}^T(\widehat Q_{k'}\widehat Q_{k'}^T)^{-1}\widehat Q_{k'},$$
assuming $\text{rank}(\widehat Q_{k'})=k'$ for the calculation to be valid.

We apply this approach to estimate the projection matrix using output from the software ADMIXTURE.

\subsubsection{Using an estimate $\widehat \Pi_{k'}$ of $\Pi_{k'}$ } 
 %\label{sec:pi}
 
Let $\widetilde \Pi_{k'}$ be $k'$ linearly independent rows chosen from $\widehat \Pi_{k'}$ (out of $m$ rows). Then, an  estimate $\widehat P_{k'}$ of $P_{k'}$ is
$$\widehat P_{k'}=\widetilde \Pi_{k'}^T(\widetilde \Pi_{k'}\widetilde \Pi_{k'}^T)^{-1}\widetilde \Pi_{k'},$$
assuming $\text{rank}(\widehat \Pi_{k'})=k'$ for the calculation to be valid. Alternatively, one might apply  the Gram-Schmidt method in which case the vectors are  orthonormal by construction and $\widehat P_{k'}=\widetilde \Pi_{k'}^T\widetilde \Pi_{k'}$.  The estimate $\widehat P_{k'}$ is independent of the choice of the $k'$ rows, provided $\text{rank}(\widehat \Pi_{k'})=k'$.

\subsubsection{Using  PCA 1}%\label{sec:CS}

We consider a PCA approach, originally due to \cite{chenstorey2015}, to estimate the space spanned by the rows of $Q$. We follow the procedure laid out in  \cite{CabrerosStorey2019}.

Let $\widehat H$ be the symmetric matrix 
$$\widehat H= \frac1mG^TG-\widehat D.$$
Since $\widehat H$ is symmetric, all eigenvalues are real and the matrix is diagonalisable. Furthermore, $\widehat H$ is a variance adjusted version of  $\frac1mG^TG$, see  \eqref{eq:Dvar}.
Let $u_1,\ldots,u_{k'}$ be $k'\le k$ orthogonal eigenvectors belonging to the $k'$ largest eigenvalues of $\widehat H$,  counted with multiplicities.  Define the $n\x k'$ matrix $U_{k'}= (u_1,\ldots,u_{k'})$ and the $n\x n$ orthogonal projection matrix 
$$\widehat P_{k'}= U_{k'}(U_{k'}^TU_{k'})^{-1}U_{k'}^T=U_{k'}U_{k'}^T$$
onto the subspace given by the span of the vectors $u_1,\ldots,u_{k'}$.

In this particular case,    convergence of $\widehat P_{k'}$ can be made precise. Define the matrix $H=4Q^T(\Sigma+\mu\mu^T)Q$. Then, $H$ is symmetric and positive semi-definite because $\Sigma$ and $\mu\mu^T$ both are positive semi-definite.  Hence, $H$ has non-negative eigenvalues.   
Furthermore, according to Lemma~\ref{thm:hatHisunbiased} in the appendix, %\Cref{sec:math},
$\widehat H$ converges    
to $H$ as $m\to\infty$.  

\begin{theorem}\label{thm:Pkconvergence}
 Assume $k'\le k$. Let $\lambda_1\ge \ldots\ge \lambda_n\ge 0$ be the eigenvalues of $H$, with corresponding orthogonal eigenvectors $v_1,\ldots,v_n$. In particular,  $\lambda_{k+1}=\ldots=\lambda_n=0$, as $Q$ has rank $k$. Let $P_{k'}$ be the orthogonal projection onto the span of $v_1,\ldots,v_{k'}$, that is,
 $$P_{k'}= V_{k'}(V_{k'}^TV_{k'})^{-1}V_{k'}^T=V_{k'}V_{k'}^T,$$
where $V_{k'}=(v_1,\ldots,v_{k'})$.

 Assume $k'=n$ or $\lambda_{k'}>\lambda_{k'+1}$, referred to as the eigenvalue condition.  Then,  $\widehat P_{k'}\to P_{k'}$  
as $m\to\infty$. If the  eigenvalue condition is fulfilled for $k'=k$, then $P_k=P$, that is, $P_k$ is the orthogonal projection onto the span of the row vectors of $Q$.  In particular, the eigenvalue condition is fulfilled for $k'=k$ if and only if $\Sigma+\mu\mu^T$ is positive definite. The latter is the case if $\Sigma$ is positive definite.
\end{theorem}

For $k'=k$, the correct row space of $Q$ is found eventually, but not $Q$ itself.  If $k'<k$, then a subspace of this row space is found, corresponding to the $k'$ largest eigenvalues.   
As the data is not mean centred, we discard the first principal component, and use the  subsequent $k'-1$ eigenvectors and eigenvalues.

\subsubsection{Using   PCA 2 (mean centred data)}
%\label{sec:genetic}

A popular approach to estimation of $\Pi$ in the admixture model is PCA based on mean centred data, or mean and variance normalised data \citep{Pritchardea2000,engelhardtstephens2010,Patterson2006}.

Let $G_1=G-\tfrac 1n GE=G(I-\tfrac 1n E)$ be the  SNP-wise mean centred genotypes, where $E$ is an $n\times n$ matrix with all entries equal to one. Following the exposition and notation in \cite{CabrerosStorey2019}, let $ G_1=U\Delta V^T$ be the SVD of $ G_1$, where $\Delta V^T$ consists of  the row-wise principal components of $ G_1$, ordered according to the singular values. 
Define 
$$S_{k'}=\begin{pmatrix} U^T_{1:(k'-1)}\\ e\end{pmatrix},$$
where $e=(1\,1\,\ldots\,1)$ is a vector with all entries one, and $U^T_{1:(k'-1)}$ contains the top $k'-1$ rows of $U^T$. Then, an estimate of the projection   is
$$ \widehat P_{k'}=S_{k'}^T(S_{k'}S_{k'}^T)^{-1}S_{k'}.$$

The   squared singular values   in the SVD decomposition of $G_1$  are the same as the eigenvalues of  
$$\widehat H_1=\frac 1m G^T_1G_1=\frac 1m\left(I-\frac 1nE\right) G^TG \left(I-\frac 1nE \right) $$
 \citep{Jolliffe2002}. We have
\begin{align}
\E[\widehat H_1] &=\frac 1m\left(I-\frac 1nE\right)\E[G^TG]\left(I-\frac 1nE\right) \nonumber\\
&=\left(I-\frac 1nE\right)(D+ 4Q^T(\Sigma+\mu\mu^T)Q)\left(I-\frac 1nE\right). \label{eq:H1}
\end{align}
Let $H_1$ denote the right hand side of \eqref{eq:H1}.

\begin{theorem}\label{thm:Pkconvergence2}
 Let $\lambda_1\ge \ldots\ge \lambda_n$ be the eigenvalues of $H_1$, with corresponding orthogonal eigenvectors $v_1,\ldots,v_n$. In particular,   $v_n=e$ and $ \lambda_n=0$. If $D$ has all diagonal entries  positive, then $\lambda_{n-1}> 0$. 
 
  Let $k'\le n$  and let $P_{k'}$ be the orthogonal projection onto the span of $v_1,\ldots,v_{k'-1},e$, that is, 
 $$P_{k'}= V_{k'}(V_{k'}^TV_{k'})^{-1}V_{k'}^T,$$
where $V_{k'}=(v_1,\ldots,v_{k'-1},e)$.
 If $k'=n$ or $\lambda_{k'}>\lambda_{k'+1}$, then $\widehat P_{k'}\to P_{k'}$  as $m\to\infty$. 
 \end{theorem}

There are no guarantees that for $k'=k$, we have   $P_k=P$ and that the difference between $\widehat B$ and $\widehat C$ converges to zero for large $m$. However, this is the case  under some extra conditions, and  appears to be the case   in many practical situations, see the Results section. %\Cref{subsec:sim}.

\begin{theorem}\label{thm:Pkconvergence3}
 Assume  $D=dI$ for some $d>0$. Furthermore, assume the vector $e$ is in the row space of $Q$ (this is, for example, the case if the admixture proportions sum to one for each individual). Then, $\lambda_k=\ldots=\lambda_{n-1}=d$, and $\lambda_n=0$.
 
 If $\Sigma+\mu\mu^T$ is positive definite, then $\lambda_{k+1}>\lambda_k$ and $P_k=P$, where $P_k$ is as in Theorem~\ref{thm:Pkconvergence}.
As a consequence, with $k'=k$ in Theorem~\ref{thm:BC}, $\widehat B-\widehat C\to0$ as $m\to\infty$.
 \end{theorem}

\subsubsection{Using   PCA 3 (mean and variance normalised data)}

Let $ G_2= W^{-1}G_1$ be the SNP mean and variance normalised genotypes, where $W$ is an $m'\times m'$ diagonal matrix with $s$-th entry being the observed standard deviation of the genotypes of SNP $s$. All SNPs for which no variation are observed are  removed, hence the number of SNPs might be smaller than the original number, $m'\le m$. Following the same procedure as for PCA 2, %in \Cref{sec:genetic}, 
let $ G_2=U\Delta V^T$ be the SVD of $ G_2$, where $\Delta V^T$ consists of the row-wise principal components of $ G_2$, ordered according to the  singular values. Define 
$$S_{k'}=\begin{pmatrix} V^T_{1:(k'-1)}\\ e\end{pmatrix},$$
where $e=(1\,1\,\ldots\,1)$, and $V^T_{1:(k'-1)}$ contains the top $k'-1$ rows of $V^T$. Then, an estimate of the projection   is
$ \widehat P_{k'}=S_{k'}^T(S_{k'}S_{k'}^T)^{-1}S_{k'}.$

We are not aware of any theoretical justification of this procedure similar to Theorem~\ref{thm:BC}, but it appears to perform well   in many practical situations, according to our simulations.  %\Cref{subsec:sim}.

\subsection{Simulation of genotype data}\label{sec:simuldetails}

We simulated genotype data from different demographic scenarios using different sampling strategies. We deliberately choose different sampling strategies to challenge the method. We first made  simple simulations that illustrate the problem of model fit as well as to demonstrate the theoretical and practical properties of the residual correlations that arise from having data from a finite number of individuals and a large number of SNPs. An overview of the simulations are given in Table \ref{tab:scenariosOverview}. 

In the first two scenarios, the ancestral allele frequencies are simulated independently for each ancestral population from a uniform distribution, $F_{s i}\sim\text{Unif}(0,1)$ for each site $s=1,\ldots,m$ and each ancestral population $i=1,\ldots,k$. In scenario 1, we simulated unadmixed individuals from three populations with either an equal or an unequal number of sampled individuals from each population. 
 In scenario 2, we simulated two ancestral populations and a population that is admixed with half of its ancestry coming from each of the two ancestral populations. 
 
 In scenario 3, we set $F_{s i}\sim\text{Unif}(0.01,0.99)$ and simulated spatial admixture in a way that resembles a spatial decline of continuous gene flow between populations living in a long narrow island. We first simulated a single population in the middle of the long island. From both sides of the island, we then recursively simulated new populations from a Balding-Nichols distribution with parameter $F_{st}=0.001$  using the R package ‘bnpsd’ \citep{ochoaStorey2019a}. In this way, each pair of adjacent populations along the island has an $F_{st}$ of 0.001. Additional details on the simulation and an schematic visualization can be found in Figure 2 of \cite{genisanders2020}.

In scenario 4, we first simulated allele frequencies for an ancestral population from a symmetric beta distribution with shape parameter 0.03, $F_{s i}\sim\text{Beta}(0.3,0.3)$, which results in an allele frequency spectrum enriched for rare variants, mimicking the human allele frequency spectrum. We then sampled allele frequencies from a bifurcating tree (((pop1:0.1,popGhost:0.2):0.05,pop2:0.3):0.1,pop3:0.5), where pop1 and popGhost are sister populations and pop3 is an outgroup. Using the Balding-Nichols distribution and the $F_{st}$ branch lengths of the tree (see Figure \ref{Fig.5}), we sampled allele frequencies in the four leaf nodes. Then, we created an admixed population with 30\% ancestry from  popGhost   and 70\% from pop2. We sampled 10 million genotypes for 50 individuals from each population except for the ghost population which was not included in the analysis, and subsequently removed sites with a sample minor allele frequency below 0.05, resulting in a total of 694,285 sites. 

In scenario 5, we simulated an ancestral population with allele frequencies from a uniform distribution $F_{s i} \sim$ \text{Unif} $(0.05, 0.95)$, from which we sampled allele frequencies for two daughter populations from a Balding Nicholds distributions with $F_{st} = 0.3$ from the ancestral population, using 'bnpsd'. We then created recent hybrids based on a pedigree where all but one founder has ancestry from the first population. The number of generations in the pedigree then determines the admixture proportions and the age of the admixture where F1 individuals have one unadmixed parent from each population and  backcross individuals have one unadmixed parent and the other F1. Double backcross individuals have one unadmixed parent and the other is a backcross. We continue to quadruple backcross with one unadmixed parent and  the other triple backcross. Note that for the recent hybrids the ancestry of the pair of alleles at each loci is no longer independent which is a violation of the admixture model.

\begin{table*}[htbp]
\centering
\caption{Overview of simulations.}
\begin{tableminipage}{\textwidth}
\begin{tabularx}{0.95\textwidth}{c|ccccc}
\hline 
\bf Scenario  & $\bm k$ & $\bm n$ & $\bm m$ &\bf Description &  $\bm F_{is}$\footnote{Ancestral allele frequencies, $i=1,\ldots,k$} \\
  \hline 
  %\hline 
  1  & 3 & 20,20,20 & $500K$& Unadmixed & $\text{Unif}(0,1)$\\
  %\hline 
  1 & 3 & 10,20,30 & $500K$ & Unadmixed & $\text{Unif}(0,1)$\\
  %\hline 
  2  & 2 & 20,20,20 & $500K$ & Admixed & $\text{Unif}(0,1)$\\
  %\hline
    2  & 2 & 10,20,30 & $500K$ & Admixed & $\text{Unif}(0,1)$\\
  %\hline
    3 &  & 500 & {\bf $100K$}\footnote{after applying MAF$>5$\% filtering, 88,082 remained.} &  Spatial with $F_{st}=0.001$& $\text{Unif}(0.01,0.99)$\\  &&&& between adjacent populations \\
   % \hline
    4 & 4 & 50,50,50,50,0\footnote{No reference samples are provided on the ghost population.}  & {\bf$10M$}\footnote{after applying MAF$>5$\% filtering, 694,285 remained.} & Ghost admixture & $\text{Beta}(0.3, 0.3)$ \\
   % \hline
    5 & 2 & 20,20,50 & $500K$ & Recent hybrids & $\text{Unif}(0.05,0.95)$\\
  \hline
\end{tabularx}
\label{tab:scenariosOverview}
\end{tableminipage}
\end{table*}

\section{Results}%\label{subsec:sim}

\subsection{Scenario 1} 
In this first set-up, we demonstrate the method using    PCA 1 only. We simulated unadmixed individuals from $k=3$ ancestral populations 
$$Q=\begin{pmatrix}
{1}_{n_{1}} & 0 &0 \\
0&  {1}_{n_{2}} & 0 \\
0& 0&  {1}_{n_{3}}
\end{pmatrix},$$
where $ {1}_{n_i}$ is a row vector with all elements  being one, and $n_1+n_2+n_3=n$. We simulated genotypes for $n=60$ individuals with sample sizes $n_1, n_2$ and $n_3$, respectively,  as detailed in the previous section. %\Cref{sec:simuldetails}. 
In Figure~\ref{Fig.1}(A), we show the residual correlation coefficients for $k'=2,3$ and plot the corresponding major PCs. For the PCA 1 approach, the first principal component does not relate to  population structure as the data is not mean centered, and we use the following $k'-1$ principal components. %which instead is captured by the principal components $2$ to $k$. %Thus, for this approach using the top $k'$ principal components corresponds to assuming that there are $k'$ ancestral populations. 
  
When assuming that there are only two populations, $k'=2$, we note that the empirical correlation coefficients appear largely consistent within each population sample, but the corrected  correlation coefficients are generally non-zero with different signs, which points to model misfit. In contrast, when assuming the correct number of populations is $k'=3$, the empirical correlation coefficients match nicely the theoretical values of $-\tfrac 1{n_i-1}$, which comply with Theorem~\ref{thm:n-1} (see  Table~\ref{tab:my_label1}). 
A fairly homogeneous pattern in the corrected correlation coefficients appears around zero across all samples. 
This is a good indication that the model fits well and that the PCA plots using principal components 2 and 3 reflex the data well.
\begin{figure}[!ht]
\centering
%\captionsetup{font={small},{labelfont=bf}}
\includegraphics[width=1\linewidth]{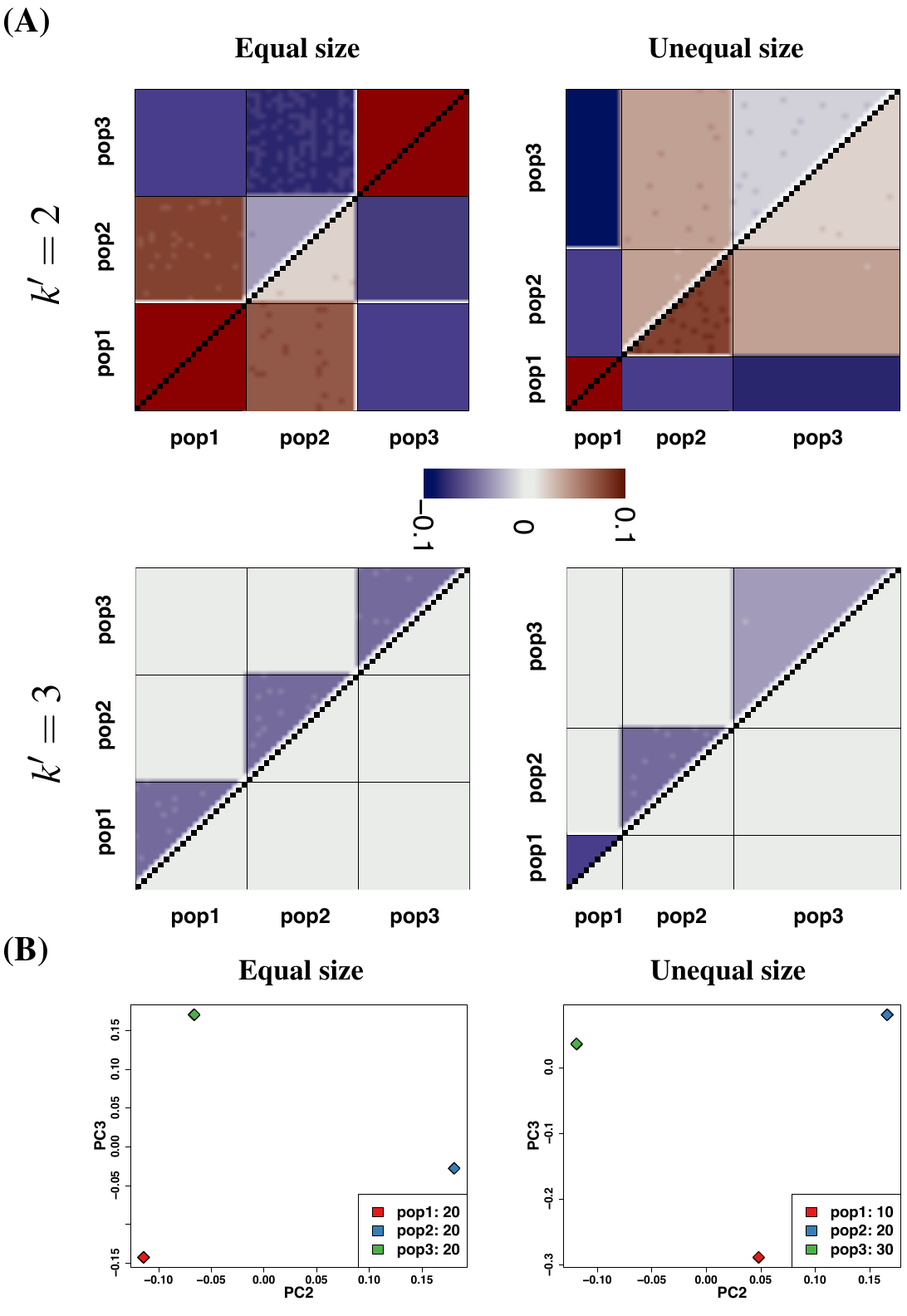}
\caption{Results for simulated Scenario 1. (A) The upper triangle in the plots shows the empirical correlation coefficients $\hat{b}$ and the lower triangle shows the corrected correlation coefficients $\hat{b}-\hat{c}$. (B) The major principal components ($k'=3$)  result in a clear separation of the three samples (all data points within each sample are almost identical). }
\label{Fig.1}
\end{figure}

\begin{table*}[!ht]
%\captionsetup{font={small},{labelfont=bf}}
\caption{The mean (standard deviation) of $\hat{b}$ and $\hat{b}-\hat{c}$ within each population using PCA 1.  }
\label{tab:my_label1}
\centering
\begin{tableminipage}{\textwidth}
\begin{tabularx}{0.95\textwidth}{c|cc|ccccc}
%\begin{tabular}{ccccc}
\hline 
\bf Scenario 1 &$\bm  k^\prime$ & $\bm n$ &&\bf pop1 &\bf pop2 &\bf pop3 \\
\hline 
& $3$ & $(20,20,20)$ &$\hat{b}\ $\footnote{The second line of $\hat{b}$ in each case shows the theoretical value obtained from the limit  in Theorem~\ref{thm:BC}.}& -0.0526 (0.0015) & -0.0526 (0.0016) & -0.0526 (0.0016) \\
&   &   && -0.0526 & -0.0526 & -0.0526 \\
& &  &$\hat{b}-\hat{c}$ & 0e-04 (0.0015) & 0e-04 (0.0016) & 0e-04 (0.0016) \\
& & $(10,20,30)$  &$\hat{b}$& -0.1111 (0.0011) & -0.0526 (0.0016) & -0.0345 (0.0016)\\
&  &&    & -0.1111 & -0.0526 & -0.0345 \\
%\hline 
&&   &$\hat{b}-\hat{c}$& 0e-04 (0.0012) & 0e-04 (0.0016) & 0e-04 (0.0016)\\
\hline 
\bf Scenario 2 &$\bm k^\prime$ & $\bm n$ &&\bf pop1 &\bf admixed &\bf pop3 \\
\hline 
& $2$  & $(20,20,20)$&$\hat{b}$ & -0.0419 (0.0015) & -0.0192 (0.0015) & -0.0420 (0.0015) \\
&  & &  & -0.0420  & -0.0193 & -0.0420 \\
& &  &$\hat{b}-\hat{c}$  & 0e-04 (0.0015) & 0e-04 (0.0015) & 0e-04 (0.0015) \\
 & &$(10,20,30)$  &$\hat{b}$& -0.0701 (0.0018) & -0.0228 (0.0014) & -0.0304 (0.0016)\\
&  & &   & -0.0701 & -0.0229  & -0.0304 \\
&& &$\hat{b}-\hat{c}$  & 0e-04 (0.0017) & 0e-04 (0.0014) &  0e-04 (0.0016) \\
\hline 
%\bf Scenario 2 &$\bm {k^\prime=2}$ &\bf pop1 &\bf admixed &\bf pop3 \\
%\hline 
%$\hat{b}$&equal  & -0.0419 (0.0015) & -0.0192 (0.0015) & -0.0420 (0.0015)\\
%& & -0.0420  & -0.0193 & -0.0420\\
%&   & -0.0701 & -0.0229  & -0.0304 \\
%\hline
%$\hat{b}-\hat{c}$&equal  & 0e-04 (0.0015) & 0e-04 (0.0015) & 0e-04 (0.0015)\\
%&unequal  & 0e-04 (0.0017) & 0e-04 (0.0014) &  0e-04 (0.0016)\\
%\hline
\bf Scenario 4 &$\bm {k^\prime}$ &$\bm n$&&\bf pop1 &\bf pop2 &\bf pop3 &\bf pop4\\
\hline 
&$3$ & $(50,50,50,50)$ & $\hat{b}$& -0.0190 (0.0015) & 0.0027 (0.0015) & -0.0204 (0.0017) &  0.0122 (0.0013) \\
 &&&$\hat{b}-\hat{c}$& 0.0009 (0.0015) & 0.0147 (0.0015) & 0e-04 (0.0017) & 0.0208 (0.0013)\\
&$4$ && $\hat{b}$& -0.0204 (0.0015) & -0.0204 (0.0015) & -0.0204 (0.0017)& -0.0204 (0.0014)\\
%\hline 
&&& $\hat{b}-\hat{c}$ & 0e-04 (0.0015) & 0e-04 (0.0015) & 0e-04 (0.0017)& 0e-04 (0.0013)\\
\hline 
\end{tabularx}
%  \label{tab:s}
\end{tableminipage}
%\end{tabular}
\end{table*}

\subsection{Scenario 2}

In this set-up we also include admixed individuals. We simulated samples from two ancestral populations and  individuals  that are a mix of the two. We then applied  all three PCA procedures and the software ADMIXTURE to the data. Specifically, we choose
$$Q=\begin{pmatrix}
 {1}_{n_{1}} & \frac{1}{2} {1}_{n_{2}} &0 \\
0& \frac{1}{2} {1}_{n_{2}} & {1}_{n_{3}} 
\end{pmatrix},$$
with  $k=2$ true ancestral populations, and  $(n_1, n_2, n_3)=(20,20,20)$ or $(n_1, n_2, n_3)=(10,20,30)$, see the previous section %\Cref{sec:simuldetails} 
for details. We analysed the data with $k^{\prime}=1,2,3$,  and obtained the correlation structure shown in  Figures~\ref{Fig.2} and \ref{Fig.3}, and Table~\ref{tab:my_label1}. The two standard approaches  PCA 2 and PCA 3 show almost identical results, hence only PCA 2 is shown in the figures. Both PCA 2 and PCA 3 use the top principal components, while PCA 1 disregards the first, hence the discrepancy in the axis labeling in Figures~\ref{Fig.2}(b) and \ref{Fig.3}/b).  %capture the population structure on all top principal components which means that $k'=3$ corresponds to the top two principal components while when assuming only one ancestral population 
For $k'=1$ none of the principal components are used and the predicted normalized genotypes is simply 0. 
All four methods show consistent results, in particular, for the correct $k'$ ($=2$), while there are smaller discrepancies between the methods for wrong $k'=1,3$. This is most pronounced for   PCA 1 and ADMIXTURE. We note that the average correlation coefficient of $\widehat b$ within each population sample comply with Theorem~\ref{thm:BC} (see Table~\ref{tab:my_label1}). 
A fairly homogeneous pattern in the corrected correlation coefficients appears around zero across all samples for $k'=2$, as in scenario 1, which shows that the model fits well. However, unlike in scenario 1 the bias for the empirical correlation coefficient is not a simple function of the sample size  (see  Table~\ref{tab:my_label1}).

In this case, and similarly in all other investigated cases, we don't find any big discrepancies between the four methods. Therefore, we only show the results of PCA 1 for which we have theoretical justification for the results.

\begin{figure*}[!ht]
\centering
%\captionsetup{font={small},{labelfont=bf}}
\includegraphics[width=0.8\linewidth]{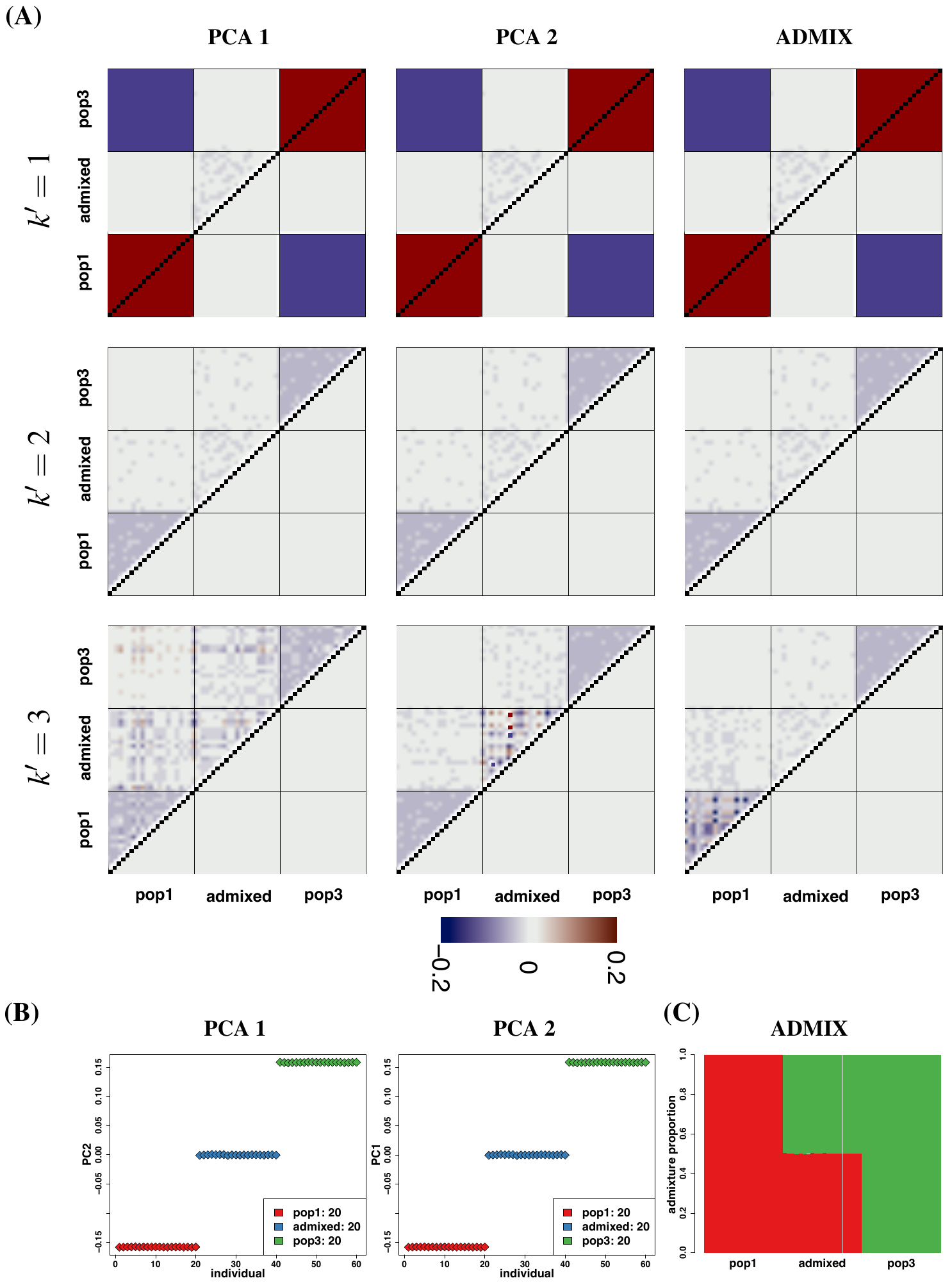}
\caption{Results for simulated Scenario 2 with equal sample sizes. (A) For each of PCA 1, PCA 2 and ADMIXTURE, the upper left triangle in  the   plots shows the empirical correlation $\hat{b}$ and the lower right triangle shows the difference $\hat{b}-\hat{c}$ with sample sizes  $(n_1, n_2, n_3)=(20,20,20)$. (B) The major principal component for the PCA based methods for $k'=2$ (in which case there is only one principal component). Individuals within each sample have the same color. (C) The estimated admixture proportions in the case of ADMIXTURE. }
\label{Fig.2}
\end{figure*}

\begin{figure*}[!ht]
\centering
%\captionsetup{font={small},{labelfont=bf}}
\includegraphics[width=0.8\linewidth]{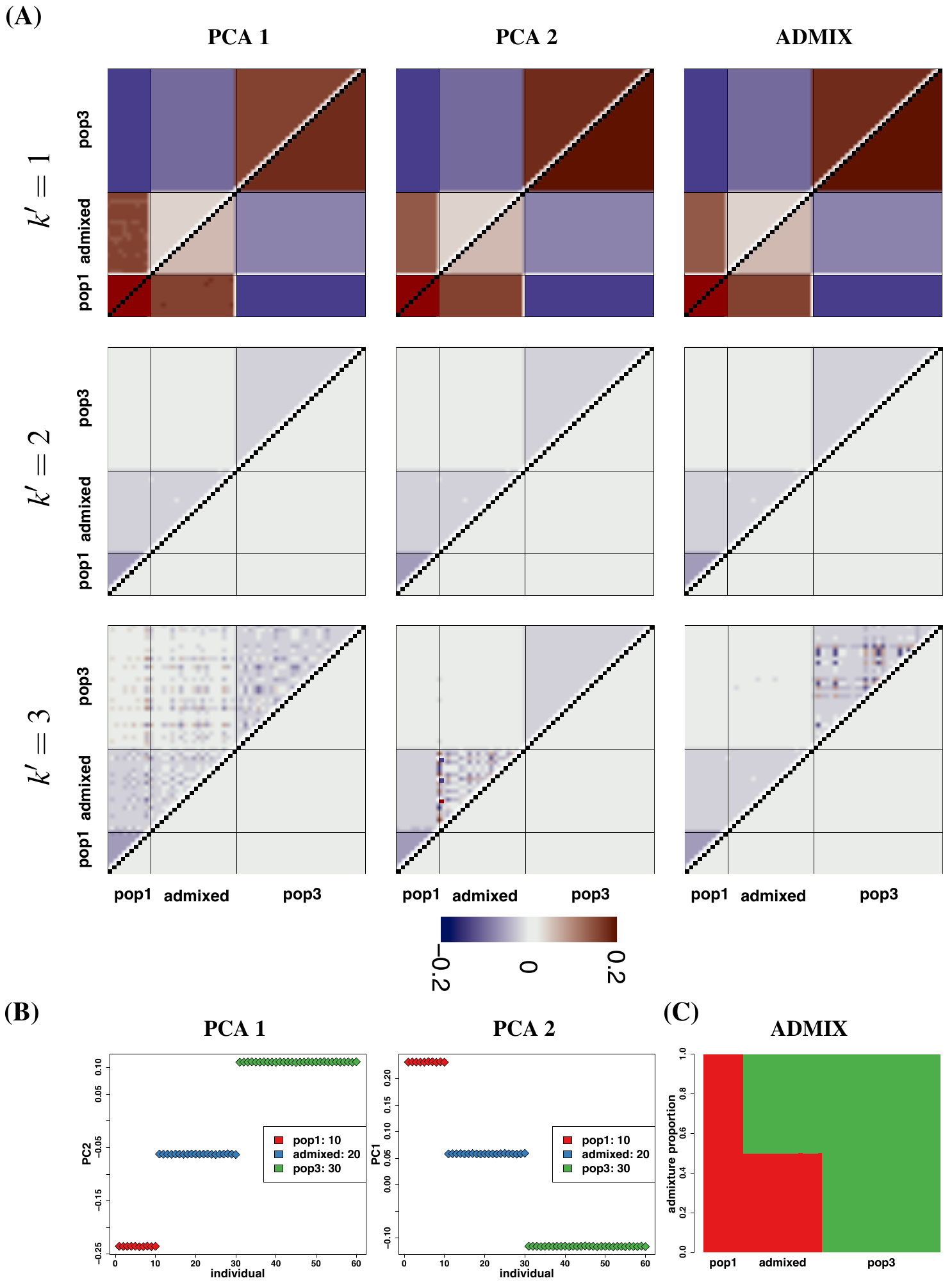}
\caption{Results for simulated Scenario 2 with unequal sample sizes. (A) For each of PCA 1, PCA 2 and ADMIXTURE, the upper left triangle in  the   plots shows the empirical correlation $\hat{b}$ and the lower right triangle shows the difference $\hat{b}-\hat{c}$ with sample sizes  $(n_1, n_2, n_3)=(20,20,20)$. (B) The major principal component for the PCA based methods for $k'=2$ (in which case there is only one principal component). Individuals within each sample have the same color. (C) The estimated admixture proportions in the case of ADMIXTURE.  }
\label{Fig.3}
\end{figure*}

\subsection{Scenario 3}
We simulated genotypes for $n=500$ individuals at $m=88,082$ sites with continuous genetic flow between individuals, thus there is not a true $k$. We analysed the data assuming $k'=2,3$, see Figure~\ref{Fig.4}. In the figure, the individuals are ordered according to the estimated proportions of the ancestral populations, hence it appears there is a color wave pattern in the empirical and the corrected correlation coefficients, see Figure~\ref{Fig.4}(A). As expected, the corrected correlation coefficients are closer to zero for $k'=3$ than $k'=2$, though the deviations from zero are still large. We thus find no support for the model for either value of $k'$. This is consistent with the plots of the major PCs, that show continuous change without grouping the data into two or three clusters, see Figure~\ref{Fig.4}(B). 

\begin{figure}[!ht]
\centering
%\captionsetup{font={small},{labelfont=bf}}
\includegraphics[width=1\linewidth]{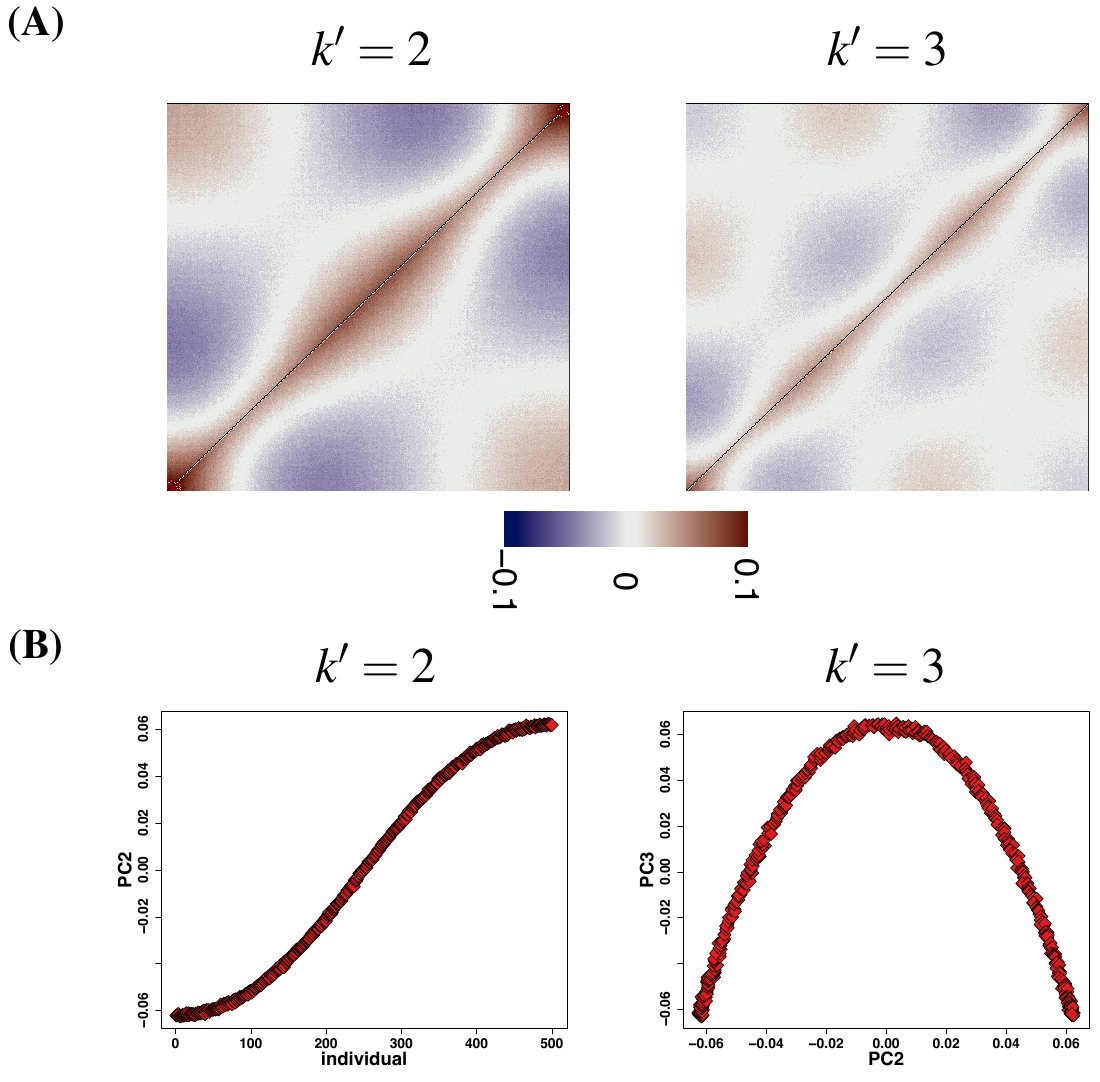}
\caption{Results for simulated scenario 3. (A) The upper triangle in the plots shows the empirical correlation $\hat{b}$ and the lower triangle shows the difference $\hat{b}-\hat{c}$. (B) The major principal components (only one in the case of $k'=2$). }
\label{Fig.4}
\end{figure}

\subsection{Scenario 4}
This case is based on the tree in Figure~\ref{Fig.5}, which include an unsampled (so-called) ghost population, popGhost.  The popGhost is   sister population to pop1.
\begin{figure}[!ht]
\centering
%\captionsetup{font={small},{labelfont=bf}}
\includegraphics[width=0.5\linewidth]{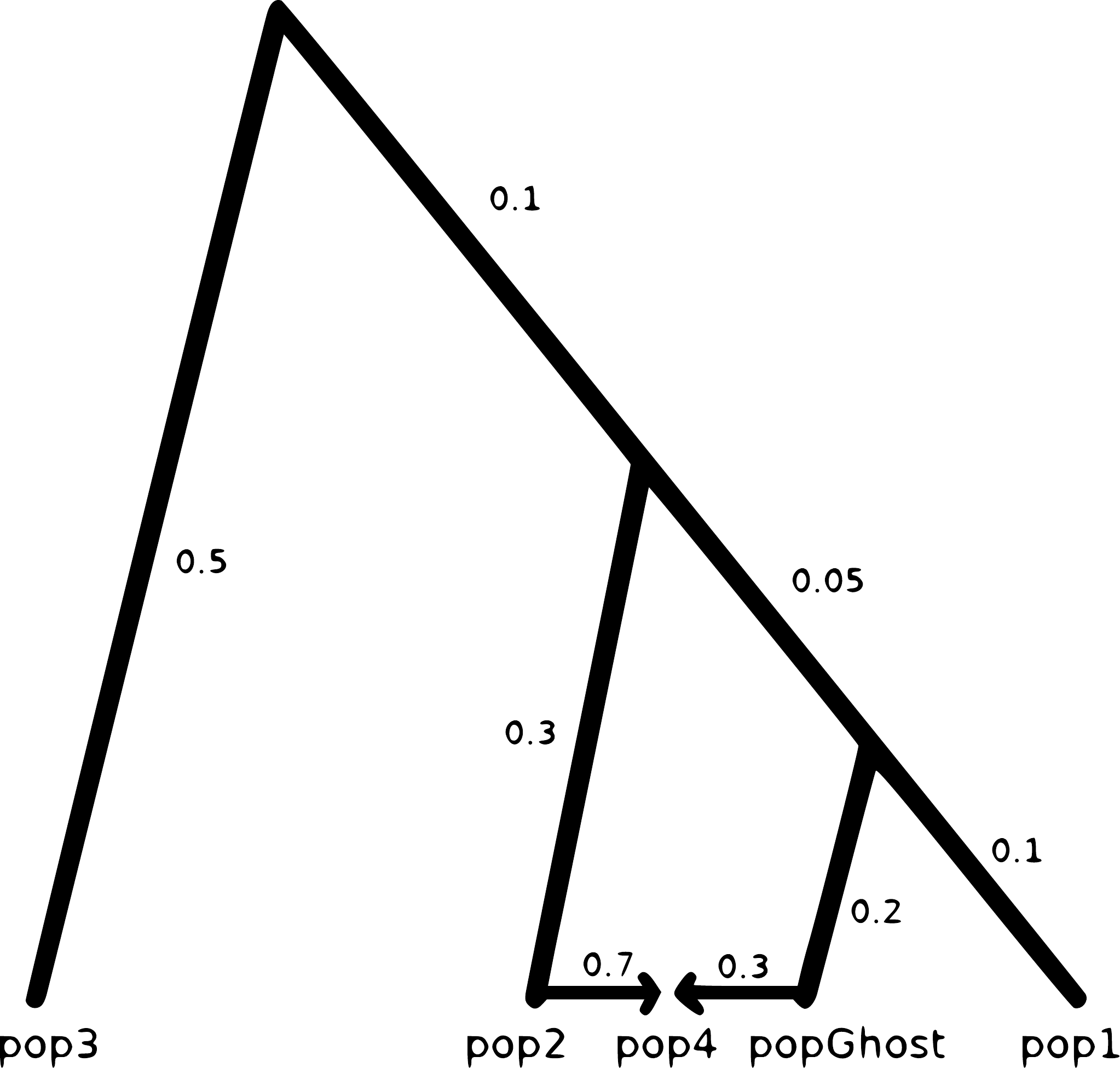}
\caption{Schematic of the tree used to simulate population allele frequencies for Scenario 4, including 5 populations: pop1, pop2, pop3, pop4 and popGhost. The pop4 population is the result of admixture between pop2 and popGhost, for which there are no individuals sampled and is therefore a ghost population. The values in the branches indicate the drift in units of $F_{ST}$. The values along the two admixture edges are the admixture proportions coming from each population. }
\label{Fig.5}
\end{figure}

We simulated genotypes for $n=200$ individuals: 150 unadmixed samples from pop1, pop2, and pop3; and 50 samples admixed with 0.3 ancestry from popGhost and 0.7 ancestry from pop2 (as pop4), as detailed in the previous section. %\ref{sec:simuldetails}. 
As there is drift between the populations and hence genetic differences, the correct $k=4$ (pop1, pop2, pop3, popGhost). This is picked up by our method that clearly shows $k'=3$ is wrong with large deviation from zero in the corrected correlation coefficients. In contrast, for $k'=4$, the corrected correlation coefficients are almost zero (Figure~\ref{Fig.6}).  

%\begin{table}[H]
%\captionsetup{font={small},{labelfont=bf}}
%\caption{Scenario 4. The mean (standard deviation) of $\hat{b}$ and $\hat{b}-\hat{c}$ within each population using PCA 1. }
%\label{tab:my_label2}
%\centering
%\begin{tabular}{cccccc}
%\hline 
% & &\bf pop1 &\bf pop2 &\bf pop3 &\bf pop4\\
%\hline 
%$\hat{b}$&$k^\prime=3$ & -0.0190 (0.0015) & 0.0027 (0.0015) & -0.0204 (0.0017) &  0.0122 (0.0013)\\
%&$k^\prime=4$ & -0.0204 (0.0015) & -0.0204 (0.0015) & -0.0204 (0.0017)& -0.0204 (0.0014)\\
%\hline 
%$\hat{b}-\hat{c}$&$k^\prime=3$ & 0.0009 (0.0015) & 0.0147 (0.0015) & 0e-04 (0.0017) & 0.0208 (0.0013)\\
%&$k^\prime=4$& 0e-04 (0.0015) & 0e-04 (0.0015) & 0e-04 (0.0017)&0e-04 (0.0013)\\
%\hline 
%\end{tabular}
%\end{table}

\begin{figure}[!ht]
\centering
%\captionsetup{font={small},{labelfont=bf}}
\includegraphics[width=0.8\linewidth]{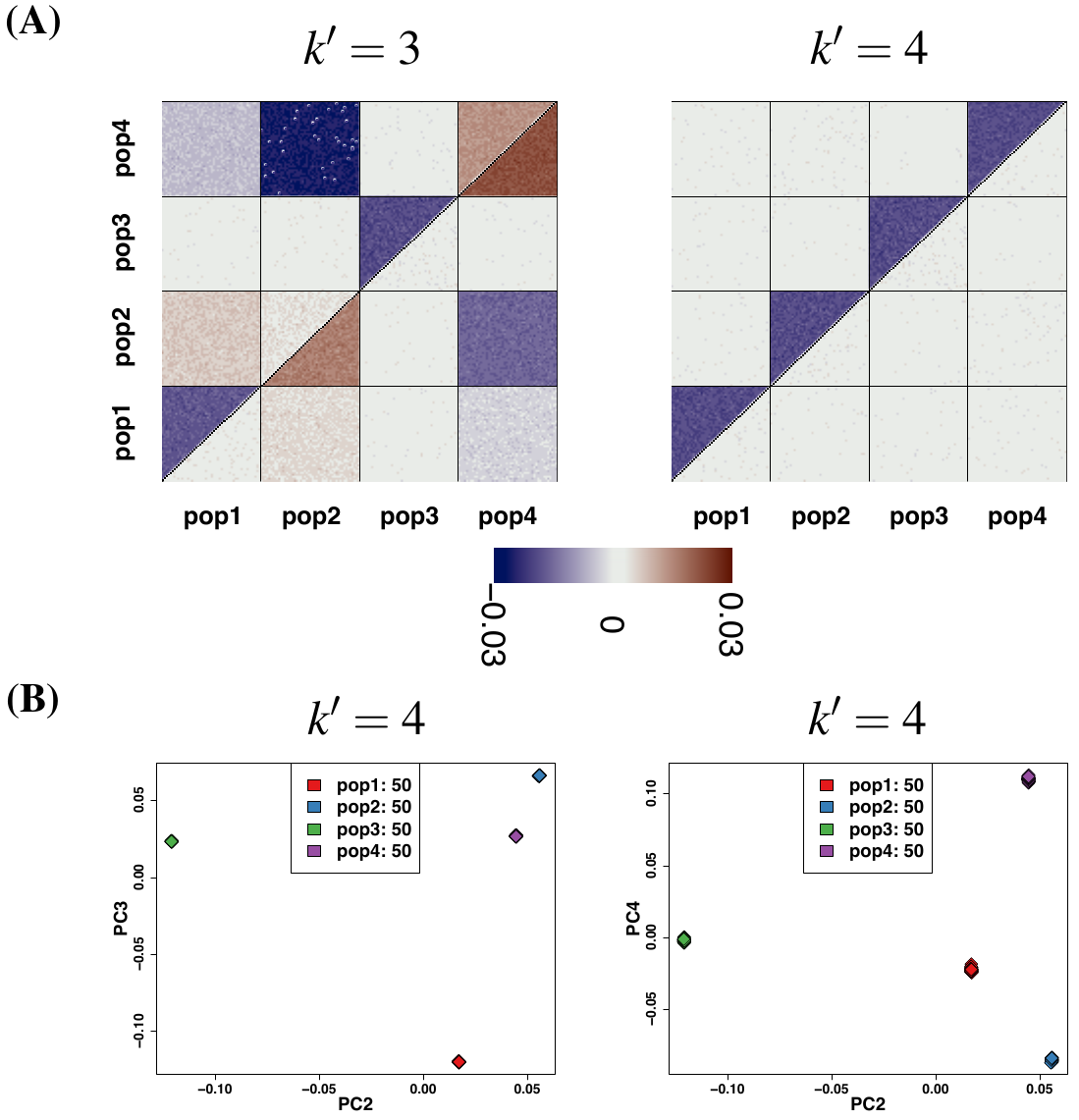}
\caption{Results for simulated scenario 4. (A) The upper triangle in the plots shows the empirical correlation $\hat{b}$ and the lower triangle shows the difference $\hat{b}-\hat{c}$. (B) The major principal components for $k'=4$, that result in a clear separation of the four samples (all data points within each sample are almost identical). }
\label{Fig.6}
\end{figure}

\subsection{Scenario 5}
In the last   example, we simulated two populations (originating from a common ancestral population) and created admixed populations by backcrossing, as detailed in the previous section. %\Cref{sec:simuldetails}. 
Thus, the model does not fulfil the assumptions of the admixture model in that the number of reference alleles are not binomially distributed, but depends on the particular backcross and the frequencies of the parental populations.

 We simulate genotypes for $n=90$ individuals at $m=500,000$ sites. There are   20 homogeneous individuals from each parental population, and 10 different individuals from each of the  different recent admixture classes. Then, we analysed the data with $k'=2$ and found the corrected correlation coefficients deviated consistently from zero, in particular for one of the parental populations (Figure~\ref{Fig.7}).  We are thus able to say the admixture model does not provide a reasonable fit.

\begin{figure}[!ht]
\centering
%\captionsetup{font={small},{labelfont=bf}}
\includegraphics[width=0.8\linewidth]{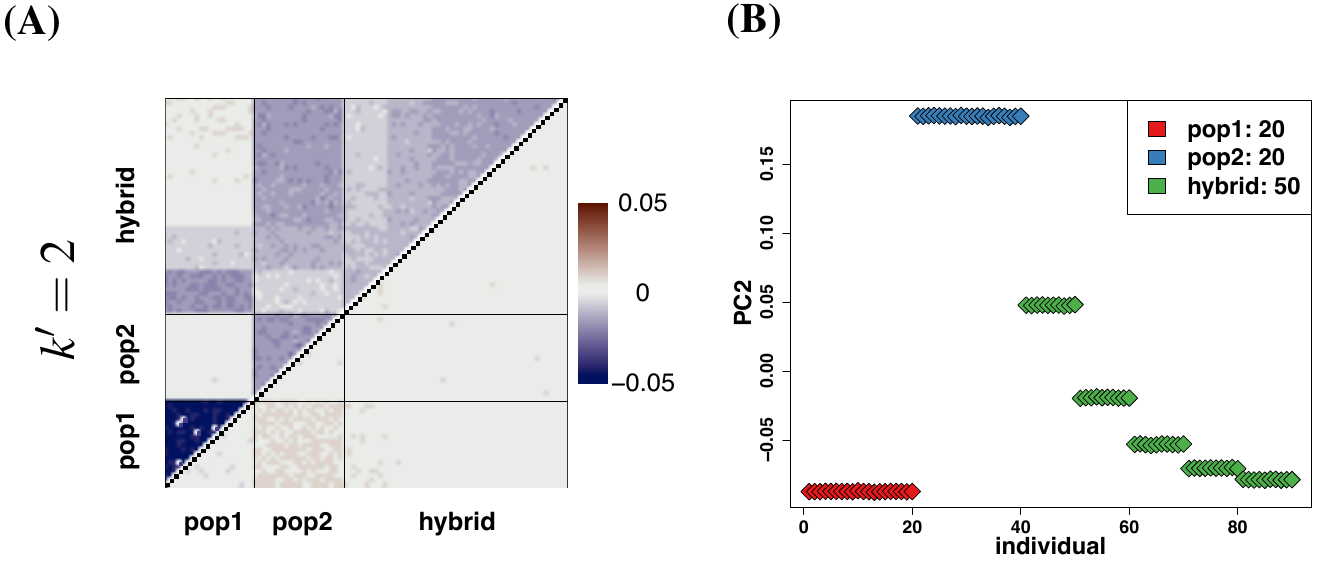}
\caption{Results for simulated scenario 5 (recent admixture). (A) The upper triangle in the plots shows the empirical correlation $\hat{b}$ and the lower triangle shows the difference $\hat{b}-\hat{c}$. (B) The major principal component for $k'=2$. }
\label{Fig.7}
\end{figure}

\subsection{Real data}%\label{subsec:1000genomesdata}

We analysed a  whole genome sequencing data set from the 1000 Genomes Project  \citep{Auton}, see also  \citet{genisanders2020} where the same data is used. It consists of data from five groups of different descent:  a Yoruba group from Ibadan,  
Nigeria (YRI), residents from Southwest US with African ancestry  (ASW), Utah residents with Northern and Western European ancestry (CEU),   a group with Mexican ancestry from Los Angeles, California (MXL), and a  group of Han Chinese from Beijing, China (CHB) with sample  sizes $108, 61, 99, 63$ and $103$, respectively, in total, $n=434$. We kept only sites present in the Human Origins SNP panel \citep{lazaridis2014}, with a  total of $m=406,279$ SNPs were left after a MAF filter of 0.05.  

We analyzed the data with $k'=3,4$. For $k'=3$, Figure~\ref{Fig.8} shows that it is not possible to explain the relationship between MXL, CEU and CHB, indicating that MXL is not well explained as a mixture of the two. For $k'=4$, the color shades of the corrected correlation coefficients  are almost negligible within each population, pointing at a contribution from a native american population. This is further corroborated in Figure~\ref{Fig.8}(D) that shows estimated proportions from the four ancestral populations using the software ADMIXTURE. 

\begin{figure}[!ht]
\centering
\captionsetup{font={small},{labelfont=bf}}
\includegraphics[width=0.8\linewidth]{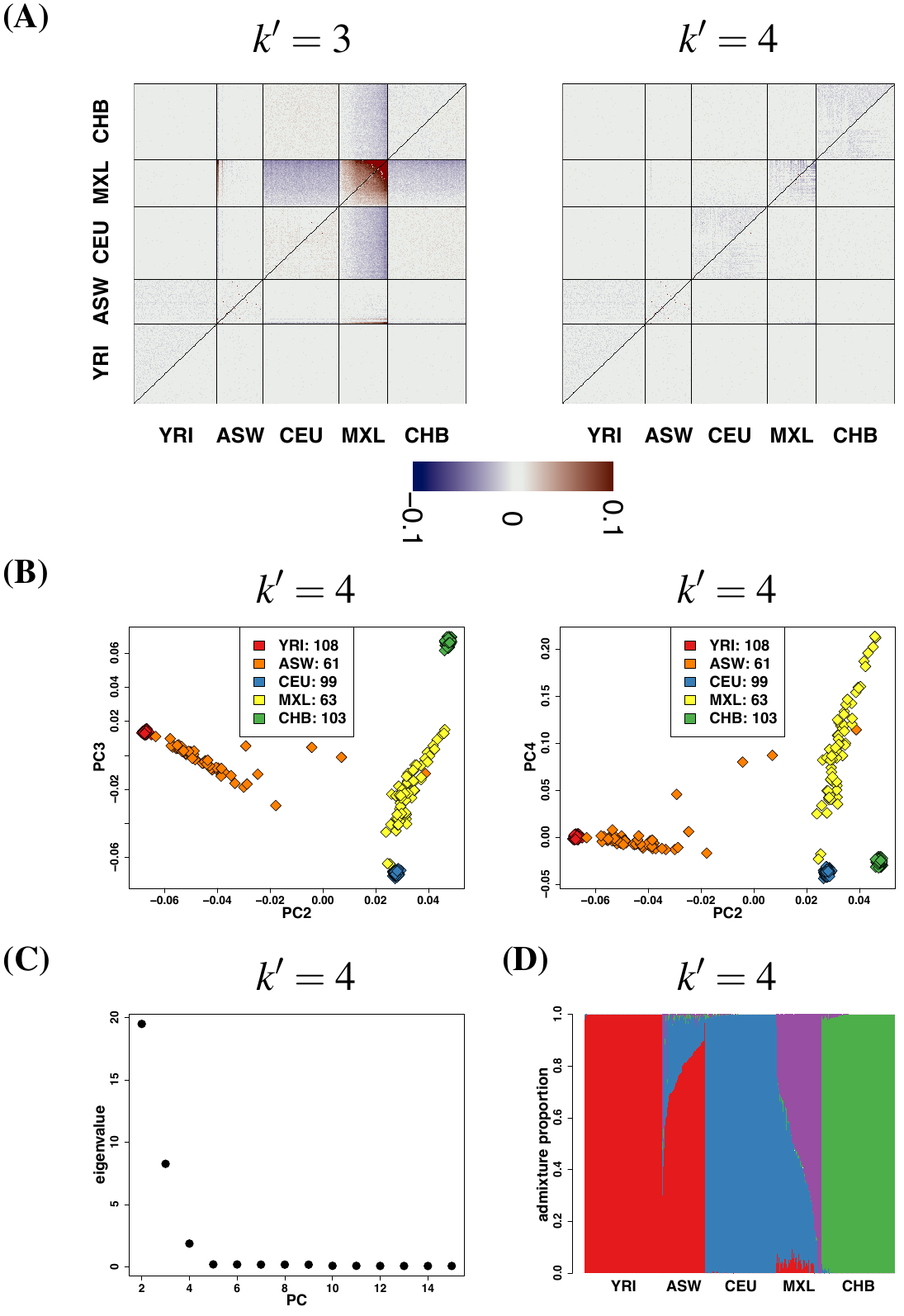}
\caption{The residual correlation coefficient, the inferred population structure and the admixture proportions of a real human data from 1000 Genomes project. (A) The upper triangle in the plots shows the empirical correlation coefficient $\hat{b}$ and the lower triangle shows the difference $\hat{b}-\hat{c}$. (B) The three major principal component for  PCA 1   for $k'=4$. (C) The eigenvalue for the first PC is removed and the eigenvalues corresponding to the remaining PCs are close to 0 after the forth PC. (D) The admixture proportions as estimated with ADMIXTURE. }
\label{Fig.8}
\end{figure}

\section{Discussion}%\label{sec:discussion}

We have developed a novel approach to assess model fit of PCA and the admixture model based on structure of the residual correlation matrix. We have shown that it performs well for simulated and real data, using a suit of different PCA methods, commonly used in the literature, and the ADMIXTURE software to estimate model parameters. By assessing the residual correlation structure visually, one is able to detect model misfit and violation of modelling assumptions.

The model fit is assessed by comparing visually two matrices of residual correlation coefficients. The theoretical and practical advantage of our approach lie in three aspects. First, our approach is computationally simple and fast. Calculation of the two residual correlation matrices and their difference is computationally inexpensive. Secondly, our approach provides a unified approach to model fitting based on PCA and clustering methods (like ADMIXTURE). In particular, it provides simple means to assess the adequacy of the chosen number of top principal components to describe the structure of the data. Assessing the adequacy by plotting the principal components against each other might lead to false confidence. In contrast, our approach exposes model misfit by plotting the difference between two matrices of the residual correlation coefficients.
Thirdly, it comes with theoretical guarantees in some cases. These guarantees are further back up by simulations  in cases, we cannot provide theoretical validity.
Finally, our approach might be adapted to work on NGS data without estimating genotypes first, but working directly on genotype likelihoods.  
 
%%%%%%%%%%%%%%%%
\section*{Data availability}
The data sets used in this study are all publicly available, including simulated  and real data. Information about the R code used to analyze and simulate  data is available at \url{https://github.com/Ginwaitthreebody/evalPCA}. The variant calls for the 1000 Genomes Project data used are publicly available at \url{ftp://ftp.1000genomes.ebi.ac.uk/vol1/ftp/release/20130502/}.

\section*{Acknowledgements}
 
The authors are supported by the Independent Research Fund Denmark (grant number:  8021-00360B) and the University of Copenhagen through the Data+ initiative.  SL acknowledges the financial support from the funding agency of China Scholarship Council. GGE and AA are supported by the Independent Research Fund Denmark (grant numbers: 8049-00098B and DFF-0135-00211B respectively).

 \appendix

 \section*{Appendix A}%\label{sec:math}

We first state the expectation and covariance matrix of $G_{s\star}$ and $\Pi_{s\star}$, respectively, under the given distributional assumptions,
\begin{align*}
\E[\Pi_{s\star}]&= \mu^TQ,\qquad
\cov(\Pi_{s\star}) =Q^T\Sigma Q,\\
\E[G_{s\star}]&=2\E[F_{s\star}Q]=2\mu^TQ,\\
\cov(G_{s\star})&=\E[\cov(G_{s\star}\mid \Pi)] +4\cov(\Pi_{s\star})
=D + 4Q^T\Sigma Q,
\end{align*}
for $s=1,\ldots,m$, where 
\begin{align*}\label{eq:D}
D&=2\E[ \diag(\Pi_{s1}(1-\Pi_{s1}),\ldots,\Pi_{sn}(1-\Pi_{sn}))],  
\end{align*}
and 
$$\E[  \Pi_{si}(1-\Pi_{si})]=\mu^TQ_{\star i} - (\mu^TQ_{\star i})^2 - (Q^T\Sigma Q)_{ii}.$$
The unconditional columns $G_{s\star}$, $s=1,\ldots,m$, of $G$ are independent random vectors by construction.

The above implies that
\begin{equation}\label{eq:star}
\frac 1m\E[G^TG]=D+4Q^T(\Sigma+\mu\mu^T)Q.
\end{equation}

Auxiliary results are  in appendix B. %\Cref{sec:proofs}.

\begin{lemma}\label{thm:unbiasedD}
    The estimator $\widehat D$  is an unbiased estimator of $D$, that is, $\E[\widehat D]=D$.
Furthermore,  it holds that $\widehat D\to D $   
as $m\to\infty$.
\end{lemma}

\begin{proof}
Conditional on $\Pi_{si}$, using binomiality, we have $\E[G_{si}(2-G_{si})\mid\Pi_{si}]= 
2\Pi_{si}(1-\Pi_{si})$, and the first result follows. For convergence, note that $G_{si}(2-G_{si})$, $s=1,\ldots,m$, unconditionally, form a sequence of iid random variables with finite variance, hence the convergence statement follows from the strong Law of Large Numbers \citep{JACOD-PROTTER}.
\end{proof}

\begin{lemma}\label{thm:hatHisunbiased}
    The estimator $\widehat H = \frac1mG^TG-\widehat D$  is an unbiased estimator of $H=4Q^T(\Sigma + \mu\mu^T)Q$, that is, $\E[\widehat H]=H$.
Furthermore,  it holds that $\widehat H\to 4Q^T(\Sigma+\mu^T\mu)Q $   
as $m\to\infty$, and
    \[
    \E\vh{\|\widehat H-4Q^T(\Sigma+\mu\mu^T)Q\|_F^2}\le \frac{16n^2}{m}.
    \]
\end{lemma}

\begin{proof}
Unbiasedness follows from \eqref{eq:star} and Lemma~\ref{thm:unbiasedD}. Consider the $(i,j)$-th entry of $\tfrac 1m G^TG$, namely, $\tfrac1m\sum_{s=1}^m G_{si}G_{sj}$. The sequence  $G_{si}G_{sj}$, $s=1,\ldots,m$, is iid with finite variance, hence $ \tfrac 1m G^TG$ converges  to $\E[G^TG]$ as $m\to\infty$ by the strong Law of Large Numbers \citep{JACOD-PROTTER}. Combined with Lemma~\ref{thm:unbiasedD} gives convergence of $\widehat H$ to $H$  as $m\to\infty$.

It remains to prove the inequality. Define
\begin{align*}
    A_{s,ij}=\begin{cases}
      G_{si}G_{sj}  - 4(Q^T(\Sigma+\mu\mu^T)Q)_{ij}
        & \text{ if } i\neq j,\\
        2G_{si}(G_{si}-1)  - 4(Q^T(\Sigma+\mu\mu^T)Q)_{ii} & 
       \text{ if } i=j.
    \end{cases} 
\end{align*}
Then,  
\begin{align*} 
(\widehat H_{ij} - \E[ \widehat H_{ij}])^2   &= \rh{\frac1m \sum_{s=1}^m  A_{s,ij}}^2 =\frac1{m^2} \sum_{s=1}^m\sum_{t=1}^m  A_{s,ij}A_{t,ij},   \\
    \|\widehat H-\E[\widehat H]\|_F^2 & =\frac1{m^2}\sum_{i=1}^n\sum_{j=1}^n \sum_{s=1}^m \sum_{t=1}^m A_{s,ij}A_{t,ij}.  
\end{align*}
Using $\E[A_{s,ij}]=0$, independence of $A_{s,ij}$ and $A_{t,ij}$ for  $s\neq t$,  and $|A_{s,ij}|\le 4$,   we have  
\begin{align*}
\E[    \|\widehat H-\E[\widehat H]\|_F^2]  
    & =\frac1{m^2}  \sum_{i=1}^n\sum_{j=1}^n \sum_{s=1}^m \E[A_{s,ij}^2]\le \frac1{m^2} 16mn^2 =\frac{16n^2}m,
\end{align*}
which proves the claim. 
\end{proof}

The convergence result  is also in \citet[theorem 2]{chenstorey2015}.   The second part provides the rate of convergence of $\widehat H$ in the $L^2$-norm.
Convergence is contingent on large $m$, rather than large $n$, and requires $m$ to increase at least like the square of $n$.

\medskip

\noindent
{\bf Proof of Theorem~\ref{thm:BC}.} Since $\widehat P_{k'}$ is assumed to be an orthogonal projection, that is, $\widehat P_{k'}^2=\widehat P_{k'}$ and $\widehat P_{k'}^T=\widehat P_{k'}$, then also the limit is an orthogonal projection,  $P_{k'}^2=P_{k'}$ and $P_{k'}^T=P_{k'}$. 

Consider the empirical covariance $\widehat B$.   Define  the variables  $T_{k'}=G(I-P_{k'})$ with $\widehat P_{k'}$ replaced by $P_{k'}$, and the empirical covariance
\begin{align*}
\widetilde B_{ij} &=\frac 1{m-1}\sum_{s=1}^m (T_{k',si}T_{k',sj}-\xbar T_{k',i} \xbar T_{k',j}) \\
&=\frac 1{m-1}\sum_{s=1}^m T_{k',si}T_{k',sj}-\frac{m}{m-1} \xbar T_{k',i} \xbar T_{k',j},
\end{align*}
defined similarly to  $\widehat B_{ij}$, with $\xbar T_{k',i}=\frac 1m\sum_{s=1}^m T_{k',si}.$  
The sequences $T_{k',si}T_{k',sj}$, $s=1,2,\ldots$, and  $T_{k',si}$, $s=1,2,\ldots$, are iid  random variables, by the distributional assumptions on $G$. Furthermore, since $ P_{k'}$ is an orthogonal projection, then $\|I-  P_{k'}\|^2_F\le n$  is bounded (Lemma~\ref{lem:normprojectionmatrix}).   Therefore, also  $T_{k',si}$ is bounded uniformly in $s,i$ by $2\sqrt{n}\le 2n$. 

Using  boundedness, independence and the strong Law of Large Numbers  \citep{JACOD-PROTTER}, 
\begin{align}\label{eq:T}
\widetilde B_{ij}&\to \E[T_{k',1i}T_{k',1j}]-\E[T_{k',1i}]\E[T_{k',1j}]=\cov(T_{k',1i},T_{k',1j}),
\end{align}
for $m\to\infty$, and $\cov(T_{k',1i},T_{k',1j})  =(I-P_{k'})(D+4 Q^T\Sigma Q)(I-P_{k'}).$ The latter equality follows  from \eqref{eq:star}.

Consider $R=G(I-\widehat P_{k'})=G(I-P_{k'})+G(P_{k'}-\widehat P_{k'})=T+G(P_{k'}-\widehat P_{k'})$. 
Hence,  
\begin{align*}
{\large|}\ {\xbar R}_{k',i}-{\xbar T}_{k',i}{\large|} &\le\frac1m \sum_{s=1}^m \sum_{i'=1}^n\sum_{j'=1}^n 2|(P_{k'}-\widehat P_{k'})_{i'j'}|\\
&=2 \sum_{i'=1}^n\sum_{j'=1}^n  |(P_{k'}-\widehat P_{k'})_{i'j'}|\to 0,
\end{align*}
as $m\to\infty$ by assumption of the theorem. It follows that ${\xbar R}_{k',i}$ converges  to $\E[T_{k',1i}]$ as $m\to\infty$.  Furthermore,
\begin{align*}
&   \frac 1{m-1}\sum_{s=1}^m  R_{k',si}R_{k',sj} - \frac 1{m-1}\sum_{s=1}^m  T_{k',si}T_{k',sj} \\
 & =  \frac 1{m-1}\sum_{s=1}^m  (T_{si} + (G(P_{k'}-\widehat P_{k'}))_{si})(T_{sj} + (G(P_{k'}-\widehat P_{k'}))_{sj}) \\
 & \quad - \frac 1{m-1}\sum_{s=1}^m  T_{k',si}T_{k',sj} \\
 % fix overbox. 
% hier
 & =  \frac 1{m-1}\sum_{s=1}^m  T_{si}(G(P_{k'}-\widehat P_{k'}))_{sj}+   \frac 1{m-1}\sum_{s=1}^m  (G(P_{k'}-\widehat P_{k'}))_{si}T_{sj}  \\
% hier
  &  \quad + \frac 1{m-1}\sum_{s=1}^m (G(P_{k'}-\widehat P_{k'}))_{si}(G(P_{k'}-\widehat P_{k'}))_{sj}.
\end{align*}
The absolute value of the first term in the last line above is bounded by 
$$\frac{4nm}{m-1}  \sum_{j'=1}^n|(P_{k'}-\widehat P_{k'})_{j'j}|,$$
 and similarly for the second term. The third is bounded by
$$\frac{4m}{m-1}  \sum_{i'=1}^n \sum_{j'=1}^n|(P_{k'}-\widehat P_{k'})_{i'i}||(P_{k'}-\widehat P_{k'})_{j'j}|.$$
All three terms converge to zero as $m\to\infty$, 
hence we conclude from \eqref{eq:T} that $\widehat B_{ij}\to \cov(T_{k',1i},T_{k',1j})$ as  $m\to\infty$.

The result for  the estimated covariance $\widehat C$ follows from convergence of $\widehat D$ and by assumption of the theorem.  The remaining part follows from the convergence of $\widehat B$ and $\widehat C$. Note that $QP=Q$, hence the second equation holds. The last statement of the theorem follows directly.   %\qed

\medskip
\noindent
{\bf Proof of Theorem~\ref{thm:n-1}.} 
Consider $T_k=G(I-P_k)=G(1-P)$, where $P=Q^T(QQ^T)^{-1}Q$ is the projection onto the row space of $Q$. Then, $T_k$ contains the residuals under multiple regression of the $m$ rows of $G$ on the $k$ rows of $Q$  \citep{box2005}. Since $e$ is in the row space of $Q$, then the sum of the residuals is zero for each $s=1,\ldots,m$: $\sum_{i=1}^n T_{k,si}=0$ (the assumption that $e$ is in the row space is   equivalent to having an intercept in the regression model) \citep{box2005}. We have, for $s=1,\ldots,m$,
\begin{align*}
0&=\var\left(\sum_{i=1}^n T_{k,si} \right) =\sum_{i=1}^n \var(T_{k,si}) + \sum_{i=1}^n \sum_{j=1,i\not=j}^n\cov(T_{k,si},T_{k,sj})  \\
&  =\sum_{i=1}^n \var(T_{k,1i}) + \sum_{i=1}^n \sum_{j=1,i\not=j}^n\cov(T_{k,1i},T_{k,1j}),
\end{align*}
since the distribution of $T_{k,si}$ is independent of $s$. From the proof of Theorem~\ref{thm:Pkconvergence}, it follows that $\widehat B$ converges to $\cov(T_{k,1\star})$ as $m\to\infty$. Hence,
\begin{align*}
\sum_{i=1}^n \widehat B_{ii}+\sum_{i=1}^n\sum_{j= 1,i\not=j }^n\widehat B_{ij} & \to \ 0,\quad\text{as}\quad m\to\infty,
\end{align*}
and the desired result follows by rearrangement. 

If $Q$ takes the given form, then the residuals under multiple regression are independent between compartments, as the projection is
$$P=\begin{pmatrix} P_1 & 0& \cdots &0\\ 0 & P_2   &\cdots &0\\ \vdots &\vdots & \ddots & \vdots \\ 0&0&\cdots& P_r\end{pmatrix},$$
where $P_\ell=Q_\ell^T(Q_\ell Q_\ell^T)^{-1}Q_\ell$ has dimension  $n_\ell\times n_\ell$. It follows that the computation above holds for each compartment. Finally, if $Q_\ell=(1 \ldots 1)$, then the distribution of the random vector $T_{k,1\star}$ is exchangeable, resulting in 
\begin{align*}
0&=\var\left(\sum_{i=1}^{n_\ell} T_{k,1i} \right) =\sum_{i=1}^{n_\ell} \var(T_{k,1i}) + \sum_{i=1}^{n_\ell} \sum_{j=1,i\not=j}^{n_\ell}\cov(T_{k,1i},T_{k,1j})\\
&=n_\ell \var(T_{k,11}) +n_\ell(n_\ell-1)\cov(T_{k,11},T_{k,12})
\end{align*}
assuming the individuals in the $\ell$-th compartment are numbered $1$ to $n_\ell$. Rearranging terms and substituting $\widehat b_{ij}$ for the moments of $T_{k,i\star}$ yields the desired result.
%\qed

\medskip
\noindent
{\bf Proof of Theorem~\ref{thm:sub}.}  Consider $T_k=G(I-P_k)=G(1-P)$, where $P=Q^T(QQ^T)^{-1}Q$ is the projection onto the row space of $Q$.  If $Q_1=(1 \ldots 1)$, then the distribution of the random variables $T_{k,11},\ldots,T_{k,1n_1}$ are exchangeable, resulting in 
\begin{align*}
0&\le \var\left(\sum_{i=1}^{n_1} T_{k,1i} \right) =\sum_{i=1}^{n_1} \var(T_{k,1i}) + \sum_{i=1}^{n_1} \sum_{j=1,i\not=j}^{n_1}\cov(T_{k,1i},T_{k,1j})\\
&=n_1\var(T_{k,11}) +n_1(n_1-1)\cov(T_{k,11},T_{k,12}).
\end{align*}
Rearranging terms and substituting $\widehat b_{ij}$ for the moments of $T_{k,i\star}$ yields the desired result.

\medskip
\noindent
{\bf Proof of Theorem~\ref{thm:Pkconvergence}.} 
The convergence statement of the theorem  is a special case of Theorem~\ref{thm:convergenceofprojection} in Appendix B. %\Cref{sec:proofs}. 
Take $A_m=\widehat H$ (that depends on the number of SNPs $m$, and the particular realization), $A=H$,  and $k=k'$ in the theorem ($k$ is used as a generic index in Theorem~\ref{thm:convergenceofprojection}). Then,  $E_kE_k^T=P_{k'}$ and $F_{m,k'}F_{m,k}^T=\widehat P_{k'}$, and the conclusion of Theorem~\ref{thm:Pkconvergence} holds.  
Convergence in Frobenius norm is equivalent to pointwise convergence (as $n$ is fixed)   $\widehat P_{k'}\to P_{k'}$ as $m\to\infty$ by definition.

If  $\Sigma+\mu\mu^T$ is positive definite, then it has rank  $k$. As $\text{rank}(Q)=k$ by assumption, it follows from Lemma~\ref{lem:rankk} that $\text{rank}(H)=k$. Consequently, there are $k$ positive eigenvalues of $H$ and $\lambda_{k+1}=0$, and the eigenvalue condition holds. %Oppositely,
Conversely, assume the  eigenvalue condition holds.   By definition $\text{rank}(H)\le k$. As $\lambda_k>\lambda_{k+1}\ge 0$ by assumption, then  also $\text{rank}(H)\ge k$ and we conclude $\text{rank}(H)= k$. It follows that the rank of $\Sigma+\mu\mu^T$ is $k$; consequently, it is positive definite.

If $k'=k=n$, then $P_k=V_kV_k^T=I$ and $P=I$ (as $k=n$), and $P_k=P$. So assume $k'=k<n$.  Since the eigenvalue condition is fulfilled, then from the above, we have $\text{rank}(Q^T(\Sigma+\mu\mu^T))=k$, and Lemma~\ref{lem:rankk} yields that the row space of $H$ and $Q$ agree.  Similarly, we have $H=V_k\diag(\lambda_1,\ldots,\lambda_k)V^T_k$ and Lemma~\ref{lem:rankk} yields that the row space of $H$ and $V_k^T$ agree. This implies the row space of $Q$ and $V_k^T$ agree.   Consequently, $P_k=Q^T(QQ^T )^{-1}Q=P$, and the statement holds. %\qed

\medskip
\noindent
{\bf Proof of Theorem~\ref{thm:Pkconvergence2}.} 
It follows trivially that $e$ is an eigenvector of $H_1$ with eigenvalue $0$. If $D$ has all entries positive, then it is positive definite and $D+4Q^T(\Sigma+\mu\mu^T)Q$ is also positive definite, hence has rank $n$. It follows from  Lemma~\ref{lem:rankk} that $H_1$ has rank $n-1$, hence $\lambda_{n-1}>0$.

 Similarly to the proof of Lemma~\ref{thm:hatHisunbiased} in Appendix B, %\Cref{sec:math}, 
 one can show $\E[\widehat H_1]=H_1$ and $\widehat H_1\to H_1$  as $m\to\infty$, where $H_1$ denotes the right hand side of \eqref{eq:H1}.  The remaining part of the theorem is proven similarly to Theorem~\ref{thm:Pkconvergence}. %\qed
 
 \medskip
\noindent
{\bf Proof of Theorem~\ref{thm:Pkconvergence3}.} 
Note that $e$ is an eigenvector of $H_1$ with eigenvalue $0$. Consider an eigenvector $v$ of $H_1$, orthogonal to $e$ with eigenvalue $\lambda$. Then, the following two equations are equivalent,
 \begin{align}
 \left(I-\frac 1nE\right)(D+ 4Q^T(\Sigma+\mu\mu^T)Q)\left(I-\frac 1nE\right)v&=\lambda v, \nonumber \\
4\left(I-\frac 1nE\right)Q^T(\Sigma+\mu\mu^T)Q\left(I-\frac 1nE\right)v&=(\lambda -d)v, \label{eq:dv} \
  \end{align}
where it is  used that $D=dI$ and $v\perp e$. It shows that $v$ is an eigenvector of $K=4(I-\tfrac 1nE)Q^T(\Sigma+\mu\mu^T)Q(I-\tfrac 1nE)$ with eigenvalue $\mu=\lambda-d$. Since $Q$ has rank $k$ and the vector $e$ is in the  space spanned by the rows of $Q$, then $Q(I-\tfrac 1nE)$ has rank $k-1$. It follows that there are at most $k-1$ positive eigenvalues of $K$, that is, at most $k-1$ eigenvalues of $H_1$ such that $\lambda>d$. Furthermore, there are precisely $k-1$ positive eigenvalues, provided $\Sigma+\mu\mu^T$ is positive definite (Lemma~\ref{lem:rankk}). The remaining eigenvalues of $K$ are zero, that is, the corresponding eigenvalues of $H_1$ are $\lambda=d$.

Assume $\Sigma+\mu\mu^T$ is positive definite, then by the above argument there precisely are $k-1$ eigenvalues of $H_1$ such that $\lambda>d$ with corresponding orthogonal eigenvectors $v_1,\ldots,v_{k-1}$. It follows from \eqref{eq:dv} that $v_1,\ldots,v_{k-1}$ are in the space   spanned by the rows of $Q(I-\tfrac 1nE)$, hence the eigenvectors are in the space spanned by the rows   of $Q$. By assumption $e$ is also in that row span. Hence,  $v_1,\ldots,v_{k-1},e$ forms an orthogonal basis of the row span of $Q$, as $Q$ has rank $k$. Thus, $P_k=P$.

\section*{Appendix  B}%\label{sec:proofs}

\begin{theorem}\label{thm:convergenceofprojection} 
    Let $A_m$ be a sequence of symmetric $n\x n$-matrices that converges to a symmetric $n\x n$-matrix $A$ in the Frobenius norm, that is $\|A_m-A\|_F\to 0$, as $m\to\infty$. Let $\lambda_1\ge \ldots\ge \lambda_n$ be the eigenvalues of $A$ (with multiplicity, and not necessarily non-negative). Let $k\le n$ be given and assume either $k=n$ or $\lambda_k>\lambda_{k+1}$. Furthermore, let $e_1,\ldots,e_k$ be  orthogonal eigenvectors corresponding to the eigenvalues $\lambda_1,\ldots,\lambda_k$, respectively, and let $f_{m,1},\ldots,f_{m,k}$ be  orthogonal  eigenvectors corresponding to the $k$ largest eigenvalues of $A_m$ (with multiplicity). Then, the orthogonal projection onto the span of $f_{m,1},\ldots,f_{m,k}$ converges to the orthogonal projection onto the span of $e_1,\ldots,e_k$ in the Frobenius norm. That is,   define $E_k=(e_1,\ldots,e_k)$ and $F_{m,k}=(f_{m,1},\ldots,f_{m,k})$, then $\|F_{m,k}F_{m,k}^T - E_kE_k^T\|_F\to 0$ as $m\to\infty$. 
\end{theorem}

\begin{proof}
If $k=n$, then $E_nE_n^T=I$ and $F_{m,n}F_{m,n}^T=I$, and the statement is trivial. Hence, assume   $k<n$. 
 Let $e_1,\ldots,e_n$ be   eigenvectors of $A$ corresponding to  eigenvalues $\lambda_1,\ldots,\lambda_n$, respectively. Let $f_{m,1},\ldots,f_{m,n}$ be the  eigenvectors of $A_m$   corresponding to the eigenvalues $\mu_{m,1}\ge \ldots\ge \mu_{m,n}$.  All eigenvectors can be asssumed to be orthonormal.    

As $\|A-A_m\|_F^2\to 0$ for $m\to\infty$, then every entry of $A_m$ converges to the corresponding entry of $A$. Consequently, the characteristic function of $A_m$ converges to that of $A$, and the eigenvalues of $A_m$ converges to those of $A$, that is, $\mu_{m,j}\to \lambda_j$ for $j=1,\ldots,n$, and $m\to\infty$.
Let  $T_m$ be such that $E_n=F_{m,n}T_m$. As $E_n$ and $F_{m,n}$ are orthogonal matrices, hence also $T_m$ is orthogonal. Applying Lemma~\ref{lem:multiplyingwithorthogonalmatrixdoesnotchangetheFrobeniusnorm} in the first and third line gives 
\begin{align*}
    \|A-A_m\|_F^2&= \|AE-A_mE_n\|_F^2= \|E\diag(\lambda_1,\ldots,\lambda_n) - A_mF_{m,n}T_m\|_F^2\\
    &= \|F_{m,n}T_m\diag(\lambda_1,\ldots,\lambda_n) - F_{m,n}\diag(\mu_{m,1},\ldots,\mu_{m,n})T_m\|_F^2\\
    &= \|T_m\diag(\lambda_1,\ldots,\lambda_n) - \diag(\mu_{m,1},\ldots,\mu_{m,n})T_m\|_F^2\\
    & =\sum_{i=1}^n \sum_{j=1}^n (\lambda_j T_{m,ij} -  \mu_{m,i}T_{m,ij} )^2
    = \sum_{i=1}^n \sum_{j=1}^n T_{m,ij}^2(\lambda_j  -  \mu_{m,i} )^2.
\end{align*}
By assumption, $\lambda_k>\lambda_{k+1}$. Hence, by convergence of eigenvalues,  for $j\le k$, $i\ge k+1$, or $j\ge k+1$, $i\le k$, we have $T_{m,ij}\to 0$ as $m\to\infty$.

Furthermore,
\begin{align*}
    E_kE_k^T - F_{m,k}F_{m,k}^T   & =\sum_{\ell=1}^{k} \rh{e_\ell e_\ell^T - f_{m,\ell}f_{m,\ell}^T } \\
     & =\sum_{\ell=1}^{k} \Big(\Big(\sum_{a=1}^n f_{m,a}T_{m,a\ell} \Big) \Big(\sum_{a=1}^n f_{m,a}T_{m,a\ell} \Big)^{\!\!T} - f_{m,\ell}f_{m,\ell}^T \Big) \\
      & =\sum_{\ell=1}^{k} \rh{\sum_{a=1}^n\sum_{b=1}^n T_{m,a\ell}T_{m,b\ell}f_{m,a}f_{m,b}^T  - f_{m,\ell}f_{m,\ell}^T } \\
   %   & =\sum_{\ell=1}^{k} \sum_{a=1}^n\sum_{b=1}^n T_{m,a\ell}T_{m,b\ell}f_{m,a}f_{m,b}^T  - \sum_{\ell=1}^{k}  f_{m,\ell}f_{m,\ell}^T\\
     &  =\sum_{a=1}^n\sum_{b=1}^n \sum_{\ell=1}^{k} T_{m,a\ell}T_{m,b\ell}f_{m,a}f_{m,b}^T  - \sum_{\ell=1}^{k}  f_{m,\ell}f_{m,\ell}^T.\\
     &=\sum_{(a,b)\in \{1,\ldots,n\}^2\setminus A_{1,k}} \Big(\sum_{i=1}^{k} T_{m,ai}T_{m,bi}\Big) f_{m,a}f_{m,b}^T \\
     &\quad +\sum_{(a,a)\in A_{1,k}} \Big(\sum_{i=1}^{k} T_{m,ai}T_{m,ai}-1\Big) f_{m,a}f_{m,a}^T,
  %   &=\sum_{a,b:(a,b)\not=(\ell,\ell),\ell=1,\ldots,k} \Big(\sum_{i=1}^{k} T_{m,ai}T_{m,bi}\Big) f_{m,a}f_{m,b}^T \\
    % &\quad +\sum_{a,b:(a,b)=(\ell,\ell),\ell=1,\ldots,k} \Big(\sum_{i=1}^{k} T_{m,ai}T_{m,bi}-1\Big) f_{m,a}f_{m,b}^T,
      \end{align*}
where $A_{i,j}=\{(a,a): i\le a \le j\}$.

From Lemma~\ref{lem:perpendicularmatricesinFrobeniusnormfromperendicularvectors}, we have $f_{m,a}f_{m,b}^T\perp f_{m,c}f_{m,d}^T$ for $(a,b)\neq (c,d)$ in the Frobenius inner product. Moreover, $\|f_{m,a}f_{m,b}^T\|_F=1$ for all $a,b$.  Hence, 
\begin{align}
   & \|E_kE_k^T-F_{m,k}F_{m,k}^T\|_F^2   \nonumber  \\
    &=\sum_{(a,b)\in \{1,\ldots,n\}^2\setminus A_{1,k}} \Big(\sum_{i=1}^{k} T_{m,ai}T_{m,bi}\Big)^{\!\!2} %\nonumber \\
   %  &\quad 
+\sum_{(a,a)\in A_{1,k}} \Big(\sum_{i=1}^{k} T_{m,ai}T_{m,ai}-1\Big)^{\!\!2}. \nonumber \\
       &=\sum_{(a,b)\in \{1,\ldots,n\}^2\setminus A_{1,n}} \Big(\sum_{i=1}^{k} T_{m,ai}T_{m,bi}\Big)^{\!\!2} +\sum_{(a,a)\in A_{k+1,n}} \Big(\sum_{i=1}^{k} T_{m,ai}T_{m,ai}\Big)^{\!\!2} \nonumber  \\
     &\quad  +\sum_{(a,a)\in A_{1,k}} \Big(\sum_{i=1}^{k} T_{m,ai}T_{m,ai}-1\Big)^{\!\!2}.\label{eq:diff}
 \end{align}
 
As noted above, $T_{m,ij}\to 0$ as $m\to\infty$ for $j\le k$, $i\ge k+1$, or $j\ge k+1$, $i\le k$. Using this and orthogonality of $T_m$ gives
$$\sum_{i=1}^{k}T_{m,a i}T_{m,b i} =\sum_{i=1}^{n}T_{m,a i}T_{m,b i} -\sum_{i=k+1}^{n}T_{m,a i}T_{m,b i} \to\begin{cases}
             0 & \text{ if } a\neq b,\\
             1 & \text{ if }a=b, \end{cases}.$$
Inserting  into \eqref{eq:diff} results in $ \|E_kE_k^T-F_{m,k}F_{m,k}^T\|_F^2 \to 0,$ as $m\to\infty$. 
\end{proof}

\begin{lemma}\label{lem:multiplyingwithorthogonalmatrixdoesnotchangetheFrobeniusnorm}
    Let $A$ be an $a\x b$ matrix. Let $U$ be a $b\x b$ orthogonal matrix and $V$ an $a\x a$ orthogonal matrix. Then,
    \begin{align*}
        \|A\|_F=\|VA\|_F=\|AU\|_F=\|VAU\|_F. 
    \end{align*}
\end{lemma}
\begin{proof} See \cite{matrix}.
\end{proof}

\begin{lemma}\label{lem:perpendicularmatricesinFrobeniusnormfromperendicularvectors} 
    Let $w,x,y,z\in\re^b$. For $a\x b$-matrices $A$ and $B$, let $\inpr{A,B}_F=\sum_{i=1}^a\sum_{j=1}^b A_{ij}B_{ij}$ be the Frobenius inner product of $A$ and $B$, and let $\inpr{\sdot,\sdot}$ be the standard inner product on $\re^b$. Then, $\inpr{wx^T,yz^T}_F=\inpr{w,y}\inpr{x,z}$. In particular, $\|wx^T\|_F=\|w\|_2\|x\|_2$ and $wx^T\perp yz^T$ if $w\perp y$ or $x\perp z$. 
\end{lemma}  
\begin{proof}
    Note that 
    \begin{align*}
        \inpr{wx^T,yz^T} = \sum_{i=1}^b\sum_{j=1}^b w_ix_j y_iz_j=\sum_{i=1}^b w_iy_i\sum_{j=1}^b x_jz_j=\inpr{w,y}\inpr{x,z}. 
    \end{align*}
    Hence, if either $w\perp y$ or $x\perp z$, then $wx^T\perp yz^T$, and $\|wx^T\|_F^2=\inpr{wx^T,wx^T}=\inpr{w,w}\inpr{x,x}=\|w\|_2^2\|x\|_2^2$, such that $\|wx^T\|_F=\|w\|_2\|x\|_2$. 
\end{proof}

\begin{lemma}\label{lem:normprojectionmatrix}
    Let $v_1,\ldots,v_\ell$ be linearly independent vectors. An orthogonal projection matrix on $\sp(v_1,\ldots,v_\ell)$ has Frobenius norm $\sqrt \ell$. 
\end{lemma}  
\begin{proof}
    We may assume that $v_1,\ldots,v_\ell$ are orthonormal. Then, we can write the projection matrix as $P=v_1v_1^T+\ldots+v_\ell v_\ell^T$. By Lemma~\ref{lem:perpendicularmatricesinFrobeniusnormfromperendicularvectors}, $v_iv_i^T\perp_F v_jv_j^T$ for $i\neq j$. So, again by Lemma~\ref{lem:perpendicularmatricesinFrobeniusnormfromperendicularvectors},  $\|P\|_F^2 = \|v_1v_1^T\|_F^2+\ldots+\|v_\ell v_\ell^T\|_F^2 = \ell$. 
\end{proof}

\begin{lemma}\label{lem:rankk}
  Let $A$ be an $a \times b$ matrix and $B$ an $b\times c$ matrix, both of rank $b$, such that $a,c\ge b$. Let $C = AB.$ Then, $C$ is of rank $b$, and the row space of $C$ coincides with the row space of $B$.
\end{lemma}

\begin{proof}First we show that $\text{rank}(C)=b$. Note that $A$ has $b$ linearly independent rows $1\le i_1 <\ldots< i_b\le b$, and $B$ has $b$ linearly independent columns $1\le j_1<\ldots<j_b\le b$. Let $\widetilde A$ and $\widetilde B$ be the $b\times b$ matrices with   $\widetilde A_{cd} = A_{i_cd}$ and $\widetilde B_{cd} = B_{c j_d}$. Then $\widetilde A$ and $\widetilde B$ are invertible matrices. Hence, also $\widetilde C=\widetilde A\widetilde B = (C_{i_aj_b})_{a,b}$ is invertible and has rank $k$. It follows that $C$ has rank $k$. 
As $\widetilde A$ is invertible, then the span of the rows of  $\widetilde A B$ is equal to the span of the rows of $B$. That is, the span of the rows of $ AB$ is equal to the span of the rows of $B$. 
\end{proof}

\bibliography{Bieb}

\end{document}